\documentclass[sigconf]{acmart}
\usepackage{booktabs}
\usepackage{xcolor}
\usepackage{url}
\usepackage{comment}
\usepackage{graphicx}
\usepackage{subfig}
\usepackage{color, colortbl}
\usepackage{listings}
\usepackage{hyperref}
\usepackage{url}

\usepackage{setspace}
\usepackage{lipsum}
\usepackage{enumitem}
\usepackage{longfbox}

\usepackage{algorithm}
\usepackage[noend]{algpseudocode}

\usepackage{amsmath}
\DeclareMathOperator*{\argmax}{arg\ max}

\usepackage{multirow}
\usepackage[mathscr]{euscript}
\usepackage{mathtools}
\usepackage[misc,geometry]{ifsym}

\usepackage{pifont}
\usepackage{hyperref}
\usepackage{cleveref}
\usepackage{url}
\usepackage{multirow}
\usepackage{booktabs}
\usepackage{amsmath}
\usepackage{listings}
\usepackage{graphicx}
\usepackage{amsmath}
\usepackage{xspace}
\usepackage{longfbox}

\definecolor{javared}{rgb}{0.6,0,0} %
\definecolor{javagreen}{rgb}{0.25,0.5,0.35} %
\definecolor{javapurple}{rgb}{0.5,0,0.35} %
\definecolor{javadocblue}{rgb}{0.25,0.35,0.75} %

\lstdefinestyle{customc}{
  belowcaptionskip=\baselineskip,
  breaklines=true,
  xleftmargin=\parindent,
  language=java,
  showstringspaces=false,
  basicstyle=\footnotesize\ttfamily,
  keywordstyle=\bfseries\color{javapurple},
  commentstyle=\itshape\blue,
 numbers=left,
 numbersep=.1cm, 
 linewidth=\columnwidth,
 xleftmargin=2.2em,
 framexleftmargin=1.1em,
 frame=single,
 float=H, 
 aboveskip=\baselineskip
}

\lstdefinestyle{CStyle}{
    backgroundcolor=\color{backgroundColour},   
    commentstyle=\color{mGreen},
    keywordstyle=\color{magenta},
    numberstyle=\tiny\color{mGray},
    stringstyle=\color{mPurple},
    basicstyle=\footnotesize,
    breakatwhitespace=false,         
    breaklines=true,                 
    captionpos=b,                    
    keepspaces=true,                 
    numbers=left,                    
    numbersep=5pt,                  
    showspaces=false,                
    showstringspaces=false,
    showtabs=false,                  
    tabsize=2,
    language=C
}

\definecolor{codegreen}{rgb}{0,0.6,0}
\definecolor{codegray}{rgb}{0.5,0.5,0.5}
\definecolor{codepurple}{rgb}{0.58,0,0.82}
\definecolor{backcolour}{rgb}{0.95,0.95,0.92}

\lstdefinestyle{braystyle}{
  commentstyle=\red,%
  keywordstyle=\blue,%
  numberstyle=\tiny\color{codegray},
  stringstyle=\color{codepurple},
  basicstyle=\scriptsize,%
  breakatwhitespace=false,         
  breaklines=true,                 
  captionpos=b,                    
  keepspaces=true,                 
  numbers=left,                    
  numbersep=4pt,                  
  showspaces=false,                
  showstringspaces=false,
  showtabs=false,                  
  tabsize=2,
  frame=None,
  language=C,
  float=H
}

\lstdefinestyle{braystyle1}{
  commentstyle=\red,%
  keywordstyle=\blue,%
  numberstyle=\tiny\color{codegray},
  stringstyle=\color{codepurple},
  basicstyle=\scriptsize,%
  breakatwhitespace=false,         
  breaklines=true,                 
  captionpos=b,                    
  keepspaces=true,                 
  numbers=left,                    
  numbersep=4pt,                  
  showspaces=false,                
  showstringspaces=false,
  showtabs=false,                  
  tabsize=2,
  frame=single,
  language=C,
  float=H
}

\newcommand*{\red}[1]{\textcolor{red}{#1}}
\newcommand*{\blue}[1]{\textcolor{blue}{#1}}

\newcommand{\ToolName}[0]{\texttt{MC$^2$}}
\newcommand{\Execution}[0]{\texttt{Monte Carlo Execution}}
\newcommand{\Graph}[0]{\texttt{PathCache}}


\copyrightyear{2022}
\acmYear{2022}
\setcopyright{acmlicensed}
\acmConference[CCS '22] {Proceedings of the 2022 ACM SIGSAC Conference on Computer and Communications Security}{November 7--11, 2022}{Los Angeles, CA, USA.}
\acmBooktitle{Proceedings of the 2022 ACM SIGSAC Conference on Computer and Communications Security (CCS '22), November 7--11, 2022, Los Angeles, CA, USA}
\acmPrice{15.00}
\acmISBN{978-1-4503-9450-5/22/11}
\acmDOI{10.1145/3548606.3560648}

\begin{document}
\fancyhead{}

\title{$MC^2$: Rigorous and Efficient Directed Greybox Fuzzing}

\author{Abhishek Shah}
\affiliation{%
  \institution{Columbia University}
  \country{}
}

\author{Dongdong She}
\affiliation{%
  \institution{Columbia University}
  \country{}
}

\author{Samanway Sadhu}
\affiliation{%
  \institution{Columbia University}
  \country{}
}

\author{Krish Singal}
\affiliation{%
  \institution{Columbia University}
  \country{}
}

\author{Peter Coffman}
\affiliation{%
  \institution{Columbia University}
  \country{}
}

\author{Suman Jana}
\affiliation{%
  \institution{Columbia University}
  \country{}
}

\begin{abstract}
Directed greybox fuzzing is a popular technique for targeted software testing that seeks to find inputs that reach a set of target sites in a program. Most existing directed greybox fuzzers do not provide any theoretical analysis of their performance or optimality. 

In this paper, we introduce a complexity-theoretic framework to pose directed greybox fuzzing as a oracle-guided search problem where some feedback about the input space (e.g., how close an input is to the target sites) is received by querying an oracle. Our framework assumes that each oracle query can return arbitrary content with a large but constant amount of information. Therefore, we use the number of oracle queries required by a fuzzing algorithm to find a target-reaching input as the performance metric. Using our framework, we design a randomized directed greybox fuzzing algorithm that makes a logarithmic (wrt. the number of all possible inputs) number of queries in expectation to find a target-reaching input. We further prove that the number of oracle queries required by our algorithm is optimal, i.e., no fuzzing algorithm can improve (i.e., minimize) the query count by more than a constant factor. 

We implement our approach in \ToolName{} and outperform state-of-the-art directed greybox fuzzers on challenging benchmarks (Magma and Fuzzer Test Suite) by up to two orders of magnitude (i.e., $134\times$) on average. \ToolName{} also found 15 previously undiscovered bugs that other state-of-the-art directed greybox fuzzers failed to find. 

\end{abstract}

\maketitle
\section{Introduction}
\label{section:introduction}
Directed greybox fuzzing is a popular technique for targeted software testing with many security applications such as bug finding, crash reproduction, checking static analyzer reports, and patch testing~\cite{cafl, aflgo, beacon, fuzzguard, hawkeye, fuzzingbook}. Given a set of target sites in a program, directed greybox fuzzers automatically search a program's input space for inputs that reach the targets. Since the input spaces of real-world programs are very large, most existing directed greybox fuzzers use evolutionary algorithms to focus their search on promising input regions identified using feedback information through instrumented program execution. For example, the fuzzers often collect feedback information about control-flow graph distance or branch constraint distance from the target and prioritize mutating inputs that are close to the target~\cite{cafl, aflgo, hawkeye}.

The designs of existing evolutionary directed fuzzers in the literature are typically guided by intuition and empirical results, but, to the best of our knowledge, they do not provide any theoretical analysis of their performance.
While strong empirical results are a necessary metric for evaluating any fuzzer design, without a rigorous theoretical understanding, it is difficult to understand the key guiding principles of fuzzer design. For example, what is the best possible (i.e., optimal) directed fuzzer? How well does a fuzzing algorithm scale with the input space size? What kind of feedback information is most useful for fuzzing? How can an algorithm best use the feedback? Can one do better than evolutionary algorithms in using this feedback information?  

In this paper, we introduce a novel computational complexity-theoretic framework to answer some of these questions and design an asymptotically optimal directed greybox fuzzing algorithm. We further demonstrate the practical advantages of our algorithm with extensive empirical results where our algorithm is up to two orders-of-magnitude faster, on average, than the state-of-the-art directed greybox fuzzers in challenging benchmarks (Magma and Fuzzer Test Suite).

\vspace{0.1cm}
\noindent\textbf{Complexity-Theoretic Framework.} To reason about an optimal fuzzer, we introduce a complexity-theoretic formulation of directed greybox fuzzing that abstracts away the specific details about the type of instrumentation and fuzzing algorithm into a unified framework. We model the task of directed greybox fuzzing as an oracle-guided search problem: to find inputs that reach the target given query access to an oracle that executes the program and reveals some information about the search space (i.e., the identity of the promising input regions) to the fuzzing algorithm.

Our framework makes no assumptions about program behaviors or input/output distributions and faithfully adheres to the design of modern greybox fuzzing. To model the lightweight fuzzing instrumentation used by practical fuzzers for collecting feedback information during a program execution, we allow the oracle to return arbitrary content with a large but constant amount of information per query. Formally, we allow the oracle to be any function $O : I \rightarrow \{0, 1\}^c$ that internally executes the instrumented target program on a constant number of inputs from an input region $I$ and returns $c$ bits of information (see Section~\ref{section:methodologymodel} for more details). 
We model the adaptive, feedback-driven paradigm used in practical fuzzers by enabling the fuzzer to arbitrarily adapt its choice of executions/input regions provided to the oracle based on the information received from the oracle in prior queries. 

\vspace{0.1cm}
\noindent\textbf{Execution Complexity.} We then introduce the notion of execution complexity, a metric to asymptotically measure the performance of any fuzzing algorithm in our framework in terms of the number of oracle queries made by the fuzzer before finding an input that reaches the target site. This metric is a very good fit for our setting as program execution and feedback generation dominate the runtime of real-world fuzzers~\cite{fullspeedfuzzing}. 
At a conceptual level, our complexity metric is similar to the query complexity used to reason about the lower bounds on potential advantages provided by quantum algorithms by separating the quantum part of the algorithm from the classical part through an oracle~\cite{qc1, qc2}. 

In our case, we want to establish lower bounds on the advantages provided by information feedback to the fuzzing algorithms and measure the corresponding execution (i.e., query) complexity. It has been shown that any search in an input space of size $N$ with a boolean oracle that only indicates whether a given test input is the desired one or not (e.g., a blackbox fuzzer) will take at least $O(N)$ queries~\cite{shaffer2012data}. In this paper, we use our framework to explore the lower bounds on the advantages of extra feedback information (i.e., large but constant per oracle query) and how to design an adaptive algorithm that can best exploit such feedback.

\vspace{0.1 cm}
\noindent\textbf{An Optimal Fuzzing Algorithm.} For designing a concrete fuzzing algorithm, we further introduce a special type of oracle called the noisy counting oracle. A noisy counting oracle takes two arbitrary input regions, approximately counts the number of inputs in each region reaching the target sites, and returns the one with higher count as the more promising one. Due to the approximate nature of the counts, we assume that the noisy oracle returns the incorrect answer with probability $p<1/2$. Later in this section, we describe how we design an approximate Monte Carlo counting algorithm to implement the noisy oracle in practice.

Given access to such a noisy oracle, we design a randomized algorithm for directed greybox fuzzing that achieves the optimal expected execution complexity $O(\frac{\log(N)}{(\frac{1}{2} - p)^2})$, up to constant factors, in an input space with $N$ inputs. We further prove that this execution complexity cannot be improved (beyond a constant factor) by any other fuzzing algorithm (evolutionary or otherwise) for any given $p<1/2$. %

As we detail in Section~\ref{section:methodologyfuzzer}, our fuzzing algorithm, at its core, uses a counting oracle to select regions with higher counts and repeatedly narrows down promising regions with binary search.
If our counting oracle was noiseless, a simple binary search based fuzzing algorithm will efficiently find target-reaching inputs. However, the noise in the counting oracle can cause the fuzzer to sometimes incorrectly select input regions with low counts. As we do not know which input regions have larger counts with full certainty, our algorithm must be robust to such noise. In this paper, we build on the noisy binary search approach introduced by Ben-Or et al.~\cite{noisybinarysearchorr}. To deal with the noise, the randomized noisy binary search algorithms~\cite{noisybinarysearchkarp, noisybinarysearchorr} maintain a set of weights for each region representing the belief of the algorithm about how likely the desired input is in that region. The algorithm %
iteratively increases the weights in promising regions based on each oracle query and prioritizes narrowing down these promising input regions.

\vspace{0.1 cm}
\noindent\textbf{Approximate Counting with Monte Carlo. } We develop a Monte Carlo algorithm for implementing the noisy counting oracle that is needed by our randomized directed greybox fuzzing algorithm. Monte Carlo random sampling is one of the popular approaches to approximate counts~\cite{karpdnf, kuldeep}. For example, to approximately count the number of people in the world who like ice cream, Monte Carlo methods will poll a random subset of individuals if they like ice cream and multiply the world's population by the ratio of people who liked ice cream in the poll to the total number of participants. The more people one polls, the more accurate the count is.

However, applying such techniques naively in our setting can result in most of the approximate counts being zero even though their true values might be small non-zero numbers.  Consider a target program with multiple branch constraints guarding the target. Suppose we wish to count the number of inputs that reach the target in an input region with $1$ million inputs and the true count of inputs reaching the target, unknown to us, is $10$. 
To efficiently approximate this count, we might execute the target program on a small number of inputs selected uniformly at random from the input region and multiply the size of the input region by the ratio of the number of inputs that reach the target in our executions to the total number of tested inputs.  The challenge, however, is that unless we generate a prohibitively large number of inputs $\sim 10^{5}$, we will estimate the count as zero because we are unlikely to find the few target-reaching inputs with a small number of randomly selected inputs.

Estimating the count of inputs reaching the target as zero can be highly detrimental to the fuzzing algorithm. In practice, most fuzzing target sites are only reachable by a small number of inputs satisfying one or multiple branch constraints. Estimating the corresponding count for any large input region as zero will degrade our fuzzing algorithm's performance significantly because it will fail to identify input regions with larger counts. To overcome this problem, we observe that even if we did not find the few inputs that reach the target, we can still compute a upper bound on the count with high confidence. 

\vspace{0.1 cm}
\noindent\textbf{Concentration Bounds. } We compute such probabilistic bounds by using concentration bounds~\cite{randomizedalgobookconcentration}, a well-studied technique to upper-bound the probability of a random variable taking a value within a given range based on the variable's mean and variance. The intuition behind such bounds is that the closer the mean (with small variance) is to the range, the higher the likelihood that the random variable will take a value within the range. For example, if we model the branch distance~\cite{branchdistance} of a branch constraint as a random variable and compute the mean and variance of the values observed at the branch during the program executions with the randomly selected inputs, we can use the concentration bounds to derive an upper bound on the likelihood of satisfying the branch constraint.

\begin{figure}[!t]
\centering
\includegraphics[scale=0.6]{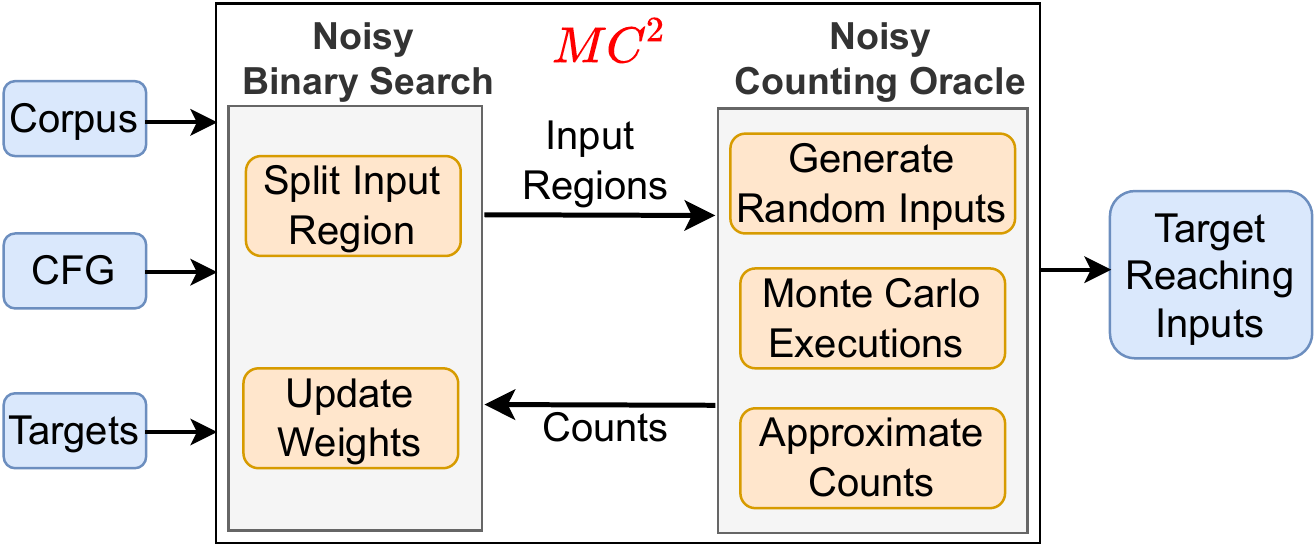}
\caption{\textbf{Workflow of \ToolName{}}}
\label{fig:workflow}
\end{figure}

One issue with these concentration bounds is that they might incur large over-approximation error that can increase the noise in the oracle, hindering the performance of the fuzzing algorithm. Therefore, we use the concentration bounds sparingly. Specifically, during a single oracle query with a set of randomly selected inputs, we first observe which branches are satisfied for which inputs and  apply probabilistic upper bounds only for the  branches that have never been satisfied for any input. For the rest, we use the empirically observed non-zero ratio of the number of inputs satisfying a branch to the total number of tested inputs.

\vspace{0.1 cm}
\noindent\textbf{Counting along Multiple Branches.} For any given path reaching a target, we approximately count the number of inputs that satisfy each branch in the path as described above and combine them to get an estimate on the count of inputs that reach the target (See Section~\ref{section:methodologycounting}). To efficiently approximate counts for nested branches, we further introduce \Execution, a lightweight form of program execution that modifies control-flows at runtime to ensure nested inner branches will be visited during the program execution even if the input did not satisfy the outer branch constraints. This technique enables us to collect information about any branch and subsequently efficiently approximate their counts with a small number of program executions. We note that even though our counting algorithm cannot provide guaranteed error bounds for arbitrary programs, our experimental results demonstrate that the method is highly effective on real-world programs.

We implement our approach in \ToolName{} (\underline{M}onte \underline{C}arlo \underline{C}ounting) shown in Figure \ref{fig:workflow} and evaluate \ToolName{} against state-of-the-art directed greybox fuzzers on challenging benchmarks with real-world programs. Our results are very promising.  In the Magma benchmark, \ToolName{} finds bugs faster by 134x, on average, and finds 16 more bugs than the next-best fuzzer. In the Fuzzer Test Suite benchmark, \ToolName{} reaches targets 77x faster, on average, and reaches 2 more targets than the next-best fuzzer. In addition, \ToolName{} found 49 previously undiscovered bugs in real-world programs, 15 more than the next-best fuzzer. We release an open source version of \ToolName{} at \url{https://hub.docker.com/r/abhishekshah212/mc2}. 

Our contributions are summarized as follows.
\begin{itemize}

\item We introduce a complexity-theoretic framework for defining directed greybox fuzzing as an oracle-guided search problem and introduce execution complexity, a metric, to measure a fuzzing algorithm's asymptotic performance. 

\item We design an asymptotically optimal randomized directed greybox fuzzer that has logarithmic expected execution complexity in the number of possible inputs. We also show that this expected execution complexity cannot be improved, up to constant terms, by any fuzzer. 

\item We develop a Monte Carlo algorithm for implementing a noisy counting oracle that can work efficiently together with our fuzzing algorithm.

\item We implement our technique in \ToolName{} and show its promise, outperforming existing state-of-the-art greybox fuzzers in challenging benchmarks (Magma and Fuzzer Test Suite) by up to two orders-of-magnitude.
\end{itemize}

\section{Methodology}
\label{section:methodology}
We first introduce a generic complexity-theoretic framework for reasoning about the best possible directed greybox fuzzer in terms of an oracle-guided search problem. We next instantiate a specific type of oracle called the noisy counting oracle and use it to build an optimal directed greybox fuzzer using noisy binary search. We then show how to implement a noisy counting oracle in practice. 

\subsection{Terminology and Notation}
Below, we provide a summary of the terminology used throughout this paper. In this paper we use the word fuzzer to refer to a directed greybox fuzzer unless otherwise noted. We denote the target program as $P$ and its large but finite-sized input space to be explored during fuzzing as $\mathbb I$. We denote the count of elements in a set with the cardinality $| \cdot |$ symbol and the size of the input space as $N=| \mathbb I |$.

\vspace{0.1 cm}
\noindent\textbf{Input Region.} 
We refer to any subset $I \subseteq \mathbb I$ as input region $I$. Since programs being fuzzed generally accept inputs as a sequence of bytes, we express, without loss of generality, the target program's input space as all possible combinations of byte values in the form of a hyperrectangle:  $\mathbb I = [0, 255]^d  $ where $d$ indicates the number of input bytes. 
Although $N=|\mathbb I|$ is finite, it is exponentially large in $d$, the number of inputs bytes to be fuzzed, which is generally bounded in real-world fuzzers for performance reasons~\cite{maxlimitaflgo, maxlimitparmesan}.

\vspace{0.1 cm}
\noindent\textbf{Control Flow Graph.} 
In the greybox setting, we have access to a target program's control flow graph: $CFG=(V, E)$, where the set of vertices $V$ represents basic blocks and the set of edges $E$ represents control-flow transitions (e.g., branches). We define a path $\pi$ as a finite sequence of edges in the CFG $\pi : E_0 \rightarrow E_1 \rightarrow ... \rightarrow E_k $, in which any two consecutive edges are adjacent.

Since we only care about paths that start from the program entry point (e.g., main), we denote the set of paths in the CFG that start from a program entry basic block and terminate at a program exit basic block as $\Pi$. As directed fuzzers are interested in reaching some target $E_T \in E$, we say that a path $\pi \in \Pi$ reaches a target if $E_T \in \pi$ and denote the subset of paths that reach target $E_T$ as $\Pi_T$.

\vspace{0.1 cm}
\noindent\textbf{Program Execution.} 
We denote executing a program $P$ on input $i$ as $P(i)$ and its corresponding path through the CFG as $\pi_{P(i)}$. We also say an input $i$ reaches the target $E_T$ if its path $\pi_{P(i)}$ reaches the target $E_T$.

\vspace{0.1 cm}
\noindent\textbf{Branch Constraint.} Each edge in our CFG corresponds to a conditional branch constraint: $c: \mathbb I \rightarrow \{0,1\}$ that takes the following normalized form: \\
\begin{tabular}{crcl}
 Constraint  &   \texttt{c($\mathbf{i}$)} & \texttt{:=} &\texttt{d($\mathbf{i}$)} $\bowtie$ \texttt{0}\\
 Predicate  &  $\bowtie$ & \texttt{:=} & \texttt{\{==,<,<=,>,>=\}}\\
  Input  &  $\mathbf{i}$ & \texttt{:=} & \texttt{[$i_1, i_2,..., i_d$]}\\
\end{tabular}\\ 

The branch distance~\cite{branchdistance} $d(\mathbf{i})$ captures the effect of all program instructions preceding the branch constraint during a program execution on input $i$.

\subsection{A Framework for Directed Greybox Fuzzing}
\label{section:methodologymodel}
While empirical measures can compare the performances of fuzzers with different designs, such measurements alone cannot tell us how close the designs are to the best possible (i.e., optimal) fuzzer. Towards this end, in this section, we introduce a complexity theoretic framework to reason about the lower bounds on the performance of an optimal fuzzer in terms of the number of target program executions. We design our framework to faithfully capture the characteristics of modern greybox fuzzers: collecting feedback information about a program execution through lightweight instrumentation and adapting their algorithms based on such feedback. 

\label{section:oracleexecutionmodel}

\vspace{0.1 cm}
\noindent\textbf{Fuzzing as Oracle-Guided Search.} %
Our framework allows the fuzzer to learn information about any bounded input region $I$ by querying an oracle. Formally, the oracle $O : I \rightarrow \{0, 1\}^c$ is any function that internally executes the program $P$ on some pre-determined constant number of inputs $i\in I$, since brute-forcing all inputs is not practical, and returns arbitrary content with a large but constant amount, $c$ bits, of information. In this context, constant means independent of input size, as is customary in complexity-theoretic analysis. 
 
Specifically, we assume that each oracle query can provide at most $c$ bits of information about a given input region either covering the entirety of the input space or some parts of it, where $c$ is some constant, capturing the information collected from a fuzzer's lightweight instrumentation. As we later demonstrate in Section~\ref{section:methodologyfuzzer}, a fuzzer can potentially use these $c$ bits to reduce the number of inputs it considers by a multiplicative factor of $\frac{1}{2^c}$ (e.g., $c=1$ bit cuts input space in half), which captures modern fuzzer's use of feedback information to reduce the number of inputs by adapting the generated inputs towards particular seeds. 

We model each oracle query as providing at most $c$ bits of information because it captures the trade-off made by real-world fuzzers: lowering the amount of instrumentation for faster execution times. Unlike symbolic/concolic execution that capture more information by collecting all the path constraints along the execution path and invoking a Satisfiability Modulo Theory (SMT) solver repeatedly for a large number of paths, modern fuzzers use lightweight instrumentation to minimize execution overhead. Hence, the upper bound of a constant number of bits of information is a natural fit for these fuzzers. 

To ensure it can reason about generic programs, our framework is distribution-free: it makes no assumptions about the target program behaviors or the types of inputs. Our framework also makes no assumptions on the type of instrumentation or the data (e.g., distance, dataflow, etc.) collected by the fuzzers as feedback. Moreover, our framework posits that fuzzers have no prior knowledge about programs and only acquire information through oracle queries, mimicking real-world fuzzing deployments that run on a large collection of programs without pre-built knowledge of program specifics. 

\vspace{0.1 cm}
\noindent\textbf{Problem Definition.} We define the task of directed greybox fuzzing as an oracle-guided search problem: given access to a program $P$, its control-flow graph $CFG$, a bounded input space $\mathbb I$, a target edge $E_T$, and an oracle $O: I \rightarrow \{0,1\}^c$, a fuzzing algorithm has to find an input $i \in \mathbb I$  such that $\pi_{P(i)} \text{ reaches the target } E_T$.

\vspace{0.1 cm}
\noindent\textbf{Execution Complexity.}
In this framework, to measure a fuzzer's performance, we analyze the number of oracle queries to solve the underlying search problem. Since the number of oracle queries directly maps to the number of program executions, up to constant factors, we define the performance metric of fuzzers in our framework as {\em execution complexity}: the number of oracle queries needed before finding an input that reaches the target. This asymptotic performance metric is well-fit for real-world fuzzers because instrumented target program executions dominate the fuzzing performance overhead~\cite{fullspeedfuzzing}. In our analysis, we ignore constant factors because they significantly depend on the hardware and we desire a hardware-independent analysis as fuzzers are run over a heterogeneous collection of hardware.

We can now reason about the best possible fuzzer with our framework. The framework provides a lower bound on the performance of any fuzzer (evolutionary or otherwise) including the best possible fuzzer, a result which we prove in Appendix~\ref{appendix:proofs}.
\begin{theorem}[Lower Bound for Any Fuzzing Algorithm]
\label{theorem:lowerbound}
Given any oracle revealing a constant $c$ bits of information per query, any fuzzing algorithm requires $\Omega(\log(N))$ execution complexity to find inputs that reach the target in an input space of size $N$. 
\end{theorem}
In the next section, we design an asymptotically optimal fuzzer. 

\vspace{0.1 cm}
\noindent\textbf{Greybox vs Blackbox Oracle.} We highlight that this lower bound is not achievable with a blackbox oracle, where an oracle query outputs a boolean value indicating if the target was reached or not for a given input. Such a blackbox fuzzing oracle, unlike the greybox one, does not provide $c=1$ bit of information as one oracle query decreases the number of inputs a fuzzer considers by only $1$, rather than by a multiplicative factor of $\frac{1}{2^1}$. Therefore, a fuzzer using a blackbox oracle (e.g., blackbox fuzzer) has $O(N)$ execution complexity ~\cite{shaffer2012data} to find target-reaching inputs in an input space of size $N$ and is asymptotically slower than a greybox fuzzer, which also matches empirical evidence~\cite{aflsmart}.

\subsection{Optimal Directed Fuzzer with Noisy Counting Oracle}
\label{section:methodologyfuzzer}

In this section, we first introduce a special kind of oracle called the noisy counting oracle that identifies which region, among two arbitrary input regions, is more promising by approximately counting the number of inputs in each region that reach the targets. Given this oracle, we then design an asymptotically optimal fuzzing algorithm that we introduce in two stages. We first describe our algorithm in an idealized setting with a noiseless counting oracle and then extend our algorithm to a realistic setting with a noisy counting oracle.

\vspace{0.1 cm}
\noindent\textbf{Noisy Counting Oracle.} We define a noisy counting oracle as taking two arbitrary input regions, approximately counting the number of inputs in each region that reach the targets, and returning $c=1$ bit of information about which region contains more inputs reaching the targets. Due to the approximate nature of the counts, we assume that the noisy oracle returns the incorrect answer with probability $p<1/2$. More formally, on input regions $I_L$ and $I_R$, the counting oracle computes the following formula:
\begin{equation}
    \label{equation:oracle}
    O(I_L, I_R) =
  \begin{cases}
        1& \text{if $C(I_L) \geq C(I_R)$} \\
        0 & \text{otherwise} \\
  \end{cases}
\end{equation}

where $C(I_L), C(I_R)$ denotes the count of inputs that reach the target in the left and right input region, respectively as defined below:
\begin{equation}
    \label{equation:count}
    C(I) = | \{ i \in I | \pi_{P(i)} \text{ reaches the target }\}|
\end{equation}

\vspace{0.1 cm}
\noindent\textbf{Optimal Deterministic Fuzzer.} 
We present a deterministic fuzzing algorithm that achieves the lower bound on execution complexity (i.e., optimal) in Theorem~\ref{theorem:lowerbound}  using a noiseless ($p=0$) counting oracle introduced above. Our fuzzer leverages binary search and splits a given input region into two input regions to use the counting oracle. However, as input regions often span multiple bytes, splitting an input region in half is ambiguous. For example, given a 2 byte input region $[0, 255] \times [0, 255]$ that can be visualized as a square, we can split the input region in half either vertically or horizontally. 

To ensure that the splitting process is unambiguous, we assign a total order to the input space, which conceptually flattens the input space into an array of size $N=| \mathbb I |$. There are many possible such orderings, but in this paper, we use lexicographic order (unless otherwise specified) which flattens byte 1, then byte 2, and so forth. In the prior example, lexicographic order conceptually flattens the input space into the following form [(0, 0), (0, 1), ..., (0, 255), (1,0), (1, 1), ... (255, 255)]. 

Our fuzzing algorithm iteratively splits an input region in half and queries the counting oracle to pick the half with higher counts, eventually finding inputs that reach the target. At each iteration, this algorithm reduces the input region size by a multiplicative factor of $\frac{1}{2^1}$ (i.e., $c=1$ bit of information) and therefore has $O(\log(N))$ execution complexity, achieving the lower bound in Theorem~\ref{theorem:lowerbound}. Algorithm~\ref{alg:fuzzer} in the Appendix illustrates this process.

\vspace{0.1 cm}
\noindent\textbf{Optimal Randomized Fuzzer.} Even though the deterministic algorithm presented earlier achieves the theoretical lower bounds, it is not practical because it depends on a noiseless counting oracle that cannot be efficiently implemented in practice.  Precisely, these types of counting problems belong to \#P, a complexity class that is at least as hard as NP~\cite{karpdnf}. Therefore, in practice, one can only hope to implement an oracle that will approximate the underlying counts. Unfortunately, our prior binary search based deterministic fuzzer design cannot work with such approximations as any error can mislead our fuzzer's selection of the input region with the largest count. We therefore design a new randomized fuzzing algorithm resilient to the noise from approximation error with techniques from the noisy binary search literature~\cite{noisybinarysearchorr, noisybinarysearchkarp, noisybinarysearchgraph}. 

In noisy binary search~\cite{noisybinarysearchorr, noisybinarysearchkarp, noisybinarysearchgraph}, algorithms perform binary search in a setting where comparisons may be unreliable, a natural fit for us because our noisy counting oracle may not reliably compare the counts between input regions due to the approximate nature of the counts. Since there is always some source of uncertainty, it is customary to develop randomized algorithms that succeed with high probability and require a small number of oracle queries in expectation (i.e., expected execution complexity), since the exact number may change each time the algorithm runs. 

We note that the expectation considers all potential behaviors of the randomized algorithm and makes no assumption about the input distribution, which is a natural fit for fuzzing because we are trying to build a practical randomized fuzzing algorithm that works well on any program. Moreover, any randomized fuzzing algorithm with some probability of success can be re-ran multiple times, so that its probability of failure exponentially decreases with more trials to a negligibly small value. Such repeated runs can also be very easily incorporated in fuzzers that run computations repeatedly in a long-running loop. 

To build our noise-resilient fuzzer, we adapt a randomized algorithm proposed in prior work by Ben-Or et al. in noisy binary search~\cite{noisybinarysearchorr}.  We state the expected execution complexity and optimality of our algorithm below, where we provide proofs in Appendix~\ref{appendix:proofs}. 

\begin{theorem}[Algorithm~\ref{alg:noisyfuzzer} Execution Complexity]
\label{theorem:complexity}
Given a noisy counting oracle that returns $c=1$ bit of information with failure probability $p < \frac{1}{2}$ per query, our fuzzing algorithm has $O((1-\delta)*\frac{\log(N)}{(\frac{1}{2} - p)^2})$ expected execution complexity to find inputs that reach the target with success probability at least $1-\delta$. 
\end{theorem}

\begin{theorem}[Algorithm~\ref{alg:noisyfuzzer} Optimality]
\label{theorem:optimality}
Given any oracle that returns a constant $c$ bits of information with failure probability $p < \frac{1}{2}$ per query, any fuzzing algorithm that succeeds with probability at least $1-\delta$ has $\Omega((1-\delta)*\frac{\log(N)}{(\frac{1}{2} - p)^2})$ expected execution complexity. 
\end{theorem}

We describe our fuzzer in Algorithm~\ref{alg:noisyfuzzer}. To better understand our fuzzing algorithm's design, we first compare it to a naive algorithm that makes multiple oracle queries at each splitting point and takes the majority result returned by those queries. Clearly, such an algorithm will be robust to a noisy oracle. However, it has $O(\frac{\log(N)*\log(\log(N))}{(\frac{1}{2} - p)^2})$ expected execution complexity ~\cite{noisybinarysearchkarp}, which is not optimal from Theorem~\ref{theorem:optimality}. By contrast, our algorithm lowers the expected execution complexity by adaptively selecting the splitting point based on each oracle query, so that there are fewer queries in total. 

In more detail, our algorithm decides where to split by maintaining a set of weights which represent its belief that an input region contains the target-reaching inputs. 
Starting from the belief that any input region is equally likely to contain inputs that reach the target, our fuzzing algorithm iteratively increases the weights in promising regions based on each oracle query and prioritizes splitting at points within these promising input regions. We note that the weight update rule can be thought of as updating the Bayesian posterior~\cite{noisybinarysearchorr} or using the multiplicative weights algorithmic paradigm~\cite{mwua}.

We note that our algorithm has two key desirable properties. First, as the failure probability $p$ in the noisy counting oracle increases, its performance degrades gracefully with a quadratic, not exponential relationship. Moreover, for a given failure probability $p$, it is optimal (i.e., cannot be improved upon) within constant factors as shown in Theorem~\ref{theorem:optimality}.   

Furthermore, our algorithm achieves logarithmic storage complexity $O(\log(N))$ in terms of the size of $\mathbb I$. A naive implementation would store a distinct weight for each input and therefore require an exponential amount of space $O(N)$. Based on the observation that many inputs share the same weight, our algorithm groups input weights together by storing the sum of their individual weights with the corresponding input region in a weight group to avoid redundancy. Each oracle query adds one additional weight group, and since there are logarithmically many oracle queries, the algorithm uses logarithmic amount of space: $O(\log(N))$ groups in expectation.

\begin{figure}[t!]
\vspace{-0.5cm}
 \begin{algorithm}[H]
\footnotesize
\caption{\small Optimal Randomized Fuzzer.} 
\label{alg:noisyfuzzer} 
\lstset{basicstyle=\ttfamily\footnotesize, breaklines=true}
\begin{tabular}{|lp{2.6in}|}\hline
\textbf{Input}:
    & \textit{$\mathbb I$} $\leftarrow$ Input Space as Array\\
    & \textit{$O$} $\leftarrow$ \textbf{Noisy} Counting Oracle from Equation~\ref{equation:oracle} and Algorithm~\ref{alg:noisycumulativecountoracle} \\
\hline
\end{tabular}
\begin{algorithmic}[1]

\State $G = \mathsf{WeightGroup()}$
\Comment{\textcolor{purple}{Initialize a single weight group: (input region, weight)}}
\State $W = [\mathsf{G}]$\Comment{\textcolor{purple}{List of weight groups}}

\While{$\text{max(W)} < \text{$\frac{1}{\sqrt{\log(|\mathbb I|)}}$}$}

        \State \textcolor{purple}{/* Find group that best splits total weight in half */}
    \State $\text{CumulativeWeight = 0}$
    \State MidIdx = 0;
    \For{\textsf{group $\in$ W}} 
        \State CumulativeWeight += group.weight
        \If{ CumulativeWeight $\geq 0.5 $}
            \State break
        \EndIf
        \State MidIdx += 1
    \EndFor

    \State 
    \State \textcolor{purple}{/* Replace W[MidIdx] with two new groups */}
    \State $I_L, I_R = \textsf{Split input region in half at W[MidIdx]}$
    \State W[MidIdx] = ($I_L$, W[MidIdx].weight)
    \State W.insert(MidIdx+1,  ($I_R$, W[MidIdx].weight)) \Comment{\textcolor{purple}{Insert group at specific index}}

    \State

    \State \textcolor{purple}{/* Perform Multiplicative Weights Update  */}
    \If{ $O(I_L, I_R) = 1 $} %
            \State \textcolor{purple}{/* Left region is more promising  */}
            \For{\textsf{idx $\in$ \{0, 1, ..., MidIdx\}}} 
                \State W[idx] *= (1-p) \Comment{\textcolor{purple}{Increase left group weights by (1-p)}}
            \EndFor
            
            \For{\textsf{idx $\in$ \{MidIdx+1, ..., |W| - 1\}}} 
                \State W[idx] *= (p) \Comment{\textcolor{purple}{Decrease right group weights by (p)}}
            \EndFor

    \Else %
            \State \textcolor{purple}{/* Right region is more promising  */}
            \For{\textsf{idx $\in$ \{0, 1, ..., MidIdx\}}} 
                \State W[idx] *= (p) \Comment{\textcolor{purple}{Decrease left group weights by (p)}}
            \EndFor
            
            \For{\textsf{idx $\in$ \{MidIdx+1, ..., |W| - 1\}}} 
                \State W[idx] *= (1-p) \Comment{\textcolor{purple}{Increase right group weights by (1-p)}}
            \EndFor

    \EndIf
    \State 
    \State \textcolor{purple}{/* Normalize such that sum of group weights is 1 */}
    \State $\text{TotalWeight = 0}$
    \For{\textsf{group $\in$ W}} 
        \State TotalWeight += group.weight
    \EndFor
    \For{\textsf{group $\in$ W}} 
        \State group.weight *= $\frac{1}{\text{TotalWeight}}$
    \EndFor
    
\EndWhile
\State $G^* = \argmax_{G \in \textsf{W}}(G.weight)$ \Comment{\textcolor{purple}{Select group with largest weight}}
\State \textbf{return} $\text{random input from $G^*$'s input region}$ \Comment{\textcolor{purple}{Inputs in $G^*$ reach the target}}

\end{algorithmic}
\end{algorithm}
\vspace{-0.4cm}
\end{figure}

\subsection{Noisy Counting Oracle through Monte Carlo Counting}
\label{section:methodologycounting}
In this section, we design a noisy counting oracle by approximately counting the number of inputs in a given input region that reach the target. Our oracle builds upon Monte Carlo counting, so we first introduce the intuitions behind this method and explain why directly using it fails in our setting. We next show how we exploit the graph structure in programs to decompose the count into a summation over individual path counts. We then show how we efficiently approximate this count from individual path counts.

\vspace{0.1 cm}
\noindent\textbf{Monte Carlo Counting.} Suppose we wish to predict the number of votes that a political candidate will receive in some country. Instead of asking every person in the country if they will vote for the candidate, Monte Carlo counting techniques efficiently approximate this count by polling a small number of randomly selected people and multiplying the number of people in the country by the ratio of people who liked the candidate in the poll to the total number of participants. Hence, Monte Carlo counting techniques trade-off accuracy for efficiency (i.e., the more people polled, the more accurate the count). In our context, a naive Monte Carlo counting strategy will be to execute the program on a small number of randomly selected inputs from the input region and multiply the input region size by the ratio of inputs that reached the target during our executions to the total number of tested inputs. 

The challenge with such a strategy is that for most input regions, the approximate count will be zero. The main problem is that the input region size $| I | $, for most real-world programs, is significantly larger than the count $C(I)$ of inputs that reach the target in the given input region, so the chance of reaching the target with a randomly selected input can be very small:  ($\frac{C(I)}{|I|} \sim 256^{-d}$). To approximate this count effectively, the naive Monte Carlo Counting strategy will require a prohibitively large number of program executions, $\sim\frac{1}{\frac{C(I)}{|I|}}$. If the counting oracle estimates the count of inputs that reach the target as zero for most input regions, our fuzzer's performance significantly degrades because it will struggle to identify the input region that contains more inputs reaching the target.

\vspace{0.1 cm}
\noindent\textbf{Exploiting CFG Structure for Counting.}
We observe that the graph structure of the CFG enables us to decompose the count $C(I)$ of inputs that reach the target in an input region $I$ into a summation of counts along any individual path $\pi \in \Pi_T$ that reaches the target. More formally: 
\begin{equation}
    \label{equation:decomposition}
    C(I) = \sum_{\pi \in \Pi_T} C_{\pi}(I)
\end{equation}
where $C_{\pi}(I)$ denote the count of inputs that reach the target along path $\pi \in \Pi_T$.

Although it is not feasible to compute this summation exactly because there can potentially be a large number of paths in a real-world program, this observation informs the design of an efficient approximation method: we can use information about how large the count is for individual paths as hints for the approximate count. 

In the next section, we describe how we efficiently approximate the count of inputs that reach the target along an individual path, for a given input region. In the subsequent section, we describe how we efficiently approximate this summation by selecting the path with the largest count. We show the entire approximate counting process in Algorithm~\ref{alg:noisycumulativecountoracle}. Although this algorithm does not have guaranteed error bounds for arbitrary programs, our experimental results in Section~\ref{sec:evaluation} demonstrate that the method is effective on real-world programs.

\subsubsection{Efficiently Approximating Individual Path Counts}
\label{section:pathcounts}
In this section, we first describe how we use uniconstraint counts to efficiently approximate the count of inputs that reach the target along an individual path in a given input region. We then describe two classes of uniconstraint counts that are challenging to compute and then how we address them.

\vspace{0.1 cm}
\noindent\textbf{Approximating Path Counts.} We wish to compute $C_{\pi}(I)$ which represents the count of inputs that reach the target $E_T$ along path $\pi: E_0 \rightarrow ... \rightarrow E_T \rightarrow ...$ in input region $I$. We observe the set of inputs that reach the target represents an intersection of multiple sets: $I_{E_1} \cap I_{E_2} \cap ... I_{E_T}$, where $I_{E_i}$ indicates the set of inputs that satisfy only the single branch constraint at edge $E_i$ in path $\pi$. The count of inputs in this intersection is strictly less than or equal to the count of inputs in any individual set $I_{E_i}$ because intersections are subsets of individual sets. Therefore, we express the count of inputs that reach the target with the following formula:  
\begin{equation}
    \label{equation:upperbound}
    C_{\pi}(I) = C(I_{E_1} \cap I_{E_2} ... I_{E_T}) \leq min(   C(I_{E_1}), C(I_{E_2}), ..., C(I_{E_T}) )
\end{equation}
where using the minimum count allows us to put an upper bound on the count of inputs along a path. 

\vspace{0.1 cm}
\noindent\textbf{Uniconstraint Counts.} 
In the above equation, $C(I_{E_i})$ represents the count of inputs that satisfy a single branch constraint at edge $E_i$ in path $\pi$, so we call them {\em uniconstraint counts}. Hence, to compute the count of inputs that reach the target along a path, we use the minimum uniconstraint count. 

We can efficiently approximate uniconstraint counts through Monte Carlo counting because individual branch constraints are likely to have a larger count of inputs that satisfy them in contrast to an intersection of multiple branch constraints, which matches empirical evidence from the symbolic execution literature~\cite{qsym}. Therefore, we are more likely to approximate them effectively with a smaller number of program executions. Formally, we use the following approximation formula for uniconstraint counts: 
\begin{equation}
    \label{equation:uniconstraintcount}
C(I_{E_i}) = |I|*r_{E_i}
\end{equation}
 where $r$ is the ratio of inputs that satisfy the branch constraint at edge $E_i$ along path $\pi$ in our random subset to the total number of inputs selected uniformly at random with replacement from input region $I$.

\vspace{0.1 cm}
\noindent\textbf{Challenges in Approximating Uniconstraint Counts. } 
Even though uniconstraint counts are more tractable to be approximated with the naive Monte Carlo counting strategy, there are two classes of uniconstraint counts that are difficult to efficiently approximate. First, some individual branch constraints may be evaluated but not satisfied in a small number of program executions \textbf{(C1)} and second, some individual branch constraints may not be evaluated at all (e.g., nested branches) in a small number of program executions \textbf{(C2)}. For these two classes of uniconstraint counts, a naive strategy will approximate their counts as zero and hence the minimum uniconstraint count that is used to approximate the count of inputs along a path, will be zero. We might choose to increase the number of program executions to handle them, but recall from Section~\ref{section:methodologymodel} that the oracle must internally execute the program a pre-determined constant number of times to avoid brute-forcing. Since we cannot know a priori how many program executions are required to effectively approximate a uniconstraint count, we describe how we address these two classes of branches below. 
\begin{figure}[t!]
\vspace{-0.5cm}
\begin{algorithm}[H]
\footnotesize
\caption{\small Monte Carlo Executions. } 
\label{alg:mcex} 
\lstset{basicstyle=\ttfamily\footnotesize, breaklines=true}
\begin{tabular}{|lrp{1.8in}|}\hline
\textbf{Input}:
    & $\textsf{P}$&$\leftarrow$ Program  \\
    &$\pi$ &$\leftarrow$ Path reaching the target \\
    & $\textsf{Inputs} $&$\leftarrow$ Set of inputs  \\
\hline
\end{tabular}
\begin{algorithmic}[1] 
    
    \State $\textsf{BranchDistances = HashMap()}$ \Comment{\textcolor{purple}{Tracks branch distances across executions}}
    \State $\textsf{BranchSatisfied = HashMap()}$ \Comment{\textcolor{purple}{Tracks if branches satisfied}}
    \For{\textsf{i $\in$ Inputs}}  
        \For{\textsf{inst $\in$ P(i)} }  \Comment{\textcolor{purple}{Execute program $P$ on input $i$}}
            \If{ \textsf{IsBranch(inst)}} \Comment{\textcolor{purple}{Check branch instruction}}
                \State \textsf{dist = GetBranchDistance(inst)}
                \State \textsf{BranchDistances[inst].Add(dist)}
                
                 \State \textsf{is\_satisfied = IsBranchSatisfied(inst)} \Comment{\textcolor{purple}{If branch satisfied 1 else 0}}
                \State \textsf{BranchSatisfied[inst] += is\_satisfied}

                \State \textsf{inst.SetBranchDirection}($\pi$) \Comment{\textcolor{purple}{Enforce runtime control-flows follow $\pi$}}
            \EndIf
            \If{ \textsf{inst.RaiseException()}} \Comment{\textcolor{purple}{Handle program exceptions}}
                \If{ \textsf{inst.ReadInvalidMem()}} 
                    \State \textsf{rand = GenerateRandom()}
                    \State \textsf{inst.SetDestination}(rand) 
                \EndIf
                \State \textbf{continue} \Comment{\textcolor{purple}{Go to next instruction}}
            \EndIf
        \EndFor
    \EndFor
    \State 
    \State \textsf{ratios = [] } \Comment{\textcolor{purple}{Compute $r$ for each branch }}
    \For{\textsf{branch$_i \in$  BranchDistances}}
        
        \State \textsf{s = BranchSatisfied[branch$_i$]}
        \If{$s > 0$}  \Comment{\textcolor{purple}{Branch satisfied at least once}}
            \State \textsf{$r_{E_i} = \frac{s}{|\textsf{Inputs}|}$}
        \Else \Comment{\textcolor{purple}{Probabilistic upper bound}}
            \State \textsf{m = Mean(BranchDistances[branch$_i$])}
            \State \textsf{v = Variance}(BranchDistances[branch$_i$])
            \State \textsf{$r_{E_i}$ = Chebyshev}(m, v, branch$_i$.predicate)  \Comment{\textcolor{purple}{Use  Table~\ref{tab:chebyshev}}}
        \EndIf
        \State \textsf{ratios.append($r_{E_i}$)}  \Comment{\textcolor{purple}{See  Equation~\ref{equation:uniconstraintcount} for interpretation of $r$}}
    \EndFor
    \State \textsf{\textbf{return} ratios}
\end{algorithmic}
\end{algorithm}
\vspace{-0.4cm}
\end{figure}

\vspace{0.1 cm}
\noindent\textbf{C1: Handling Evaluated but Unsatisifed Branches. } 
Although some branches might be evaluated in a small number of executions, they might never be satisfied for any tested input, so we will approximate their uniconstraint count as zero (i.e., $r=0$). We overcome this by computing a probabilistic upper bound on these branches uniconstraint counts using a concentration bound called Chebyshev's inequality~\cite{chebyshev} that sets $r$, the likelihood of satisfying the branch constraint, based on the sample mean and variance of observed branch distances~\cite{branchdistance} during the program executions. Such probabilistic upper bounds are a natural fit in our setting since our uniconstraint counts are themselves upper bounds of the counts of all inputs taking a path.

Specifically, we model branch distance $d(i)$ as a random variable $X$ with mean $\mu$ and variance $\sigma$ to use Chebyshev's inequality. Table~\ref{tab:chebyshev} shows Chebyshev's inequality for any form of branch constraint. It also models logical operators AND and OR as seen by the equality and inequality. Note that Chebyshev's inequality assumes nothing about the distribution except that the mean and variance are finite, which holds for programs run on finite bit-precision hardware. Therefore, it applies for any program behavior and variable type (floats, integers, etc). 

Moreover, we can be confident in our probabilistic upper bounds because the approximation error of the sample mean and variance decreases exponentially fast in the number $k$ of program executions $\sim \frac{1}{e^k}$ for any random variable (i.e., Chernoff–Hoeffding inequality ~\cite{randomizedalgobookconcentration, chernoffsamplevariance}), and therefore, we can derive high quality approximations with a small number of executions. In addition, the quality of the approximation error and how likely it is to occur can be quantified and controlled through $(\delta, \epsilon)$ bounds as shown in the Probably Approximately Correct (P.A.C.) framework~\cite{pacvaliant}.  Note that since these upper bounds can potentially result in a large over-approximation error from the true count, we use them only for branches that were never satisfied in our tested inputs. For the rest, we use the empirically observed non-zero ratio of the number of inputs satisfying the branch to the total number of tested inputs. 
\begin{table}[!t]
    \caption{\small Rules for computing an upper bound on $r$ from Equation~\ref{equation:uniconstraintcount}. We model the branch distance $d(i)$ as a random variable $X$ with mean $\mu$ and variance $\sigma$. $h$ represents the smallest positive number for the data type of $d(i)$ (i.e., for integers, $h=1$).}
    \vspace{-0.2cm}
   \footnotesize
    \centering
    \renewcommand{\arraystretch}{1.5}
    \setlength{\tabcolsep}{1.5pt}
    \begin{tabular}{| l | l |}
    \hline
    Branch constraint & Rule to compute $r$ \\
    \hline
    $d(i) \leq 0$ &
    $Pr(X \leq 0) \leq \frac{\sigma}{\sigma + \mu^2}$ \\
    
    \hline
    
    $d(i) < 0$ & $Pr(X \leq -h) = Pr(X + h \leq 0)$ \\
    
    \hline
    
    $d(i) \geq 0$ &
    $Pr(X \geq 0) \leq  \frac{\sigma}{\sigma + \mu^2}$\\
    
    \hline
    $d(i) > 0$ & $Pr(X \geq h) = Pr(X -h \geq 0)$ \\
    \hline

    $d(i) = 0$ & 
    $Pr(X \geq 0 \land X \leq 0) = \min(Pr(X \geq 0), Pr(X \leq 0))$\\   
    \hline
    $d(i) \neq 0$ & 
    $Pr(X > 0 \lor X < 0) = Pr(X > 0) + Pr(X < 0)$\\

    \hline
    \end{tabular}
    \label{tab:chebyshev}
\vspace{-2mm}
\end{table}

\vspace{0.1 cm}
\noindent\textbf{C2: Handling Unevaluated Nested Branches. }
In our small number of program executions, some branches may not be evaluated at all because they are nested and since we have no information about such unevaluated nested branches, we will approximate their uniconstraint counts as zero. Inspired by prior work in malware analysis~\cite{forcedexecutionother, forcedexecutionpmp, forcedexecutionxforce}, we instead design a new form of execution called \Execution \ that ensures that any input will visit and evaluate all nested inner branches, even if the prior outer constraints along the way are unsatisfied. Hence, in a single execution, we will visit and evaluate all branches together and therefore effectively approximate the uniconstraint counts for all branches, even with a small number of program executions. 

Given a path $\pi \in \Pi_T$ consisting of a set of desired branches reaching the target $T$, \Execution \ modifies control-flows at runtime to always visit these branches, irrespective of the input. Note that \Execution \ does not necessarily change the original execution path: if an input satisfies all branch constraints in $\pi$, \Execution \ behaves as an original execution. However, it can deviate from the original execution path if the input does not satisfy any one of these branch constraints. Even if \Execution \ deviates from the original execution path for an input, the input goes through the same computation as if it was a valid input. Hence, \Execution \ always preserve the sequential ordering of computation. 

To ensure that it will always visit the desired set of  branches, \Execution \ must handle program exceptions. For example, the program can make an out of bounds memory access if the input controls the index of an array. We handle them by advancing the instruction pointer and if the program exception was raised by an invalid memory read, we also replace the destination with a uniformly random value to avoid bias in the computed values between individual executions. 

This design can increase the set of possible values for the destination, but this is a natural fit since we use upper bounds for counts. Even though this design loses dependencies across memory reads and writes, it has low overhead in contrast to prior work that attempts to preserve these dependencies~\cite{forcedexecutionother, forcedexecutionxforce}. In Section~\ref{sec:evaluationrq3} and Appendix~\ref{appendix:crashrate}, we run experiments to better understand this overhead. Algorithm~\ref{alg:mcex} depicts the entire process of \Execution \ on a set of inputs.

\subsubsection{Efficiently Approximating The Summation}
\label{section:sumcounts}

The method described in the prior section only deals with a single path reaching the target. To handle multiple paths, we need to sum each path's count to get a total count as mentioned in Equation~\ref{equation:decomposition}. The challenge, however, is that although we can efficiently approximate counts for a single path, performing this procedure for each path at each oracle query quickly becomes computationally intractable if there are a large number of paths reaching the target.

Instead of computing this sum through contributions from each individual path count at each oracle query, we approximate the sum by only including contributions from a single path with the largest corresponding count. We select the largest-count path
as its count best preserves the sum compared to any other single path. 
However, we do not apriori know which path has the largest count, so we initially spend some computation approximating each path's individual count, amortizing this cost over subsequent oracle queries. Hence, on the fuzzer's first oracle query, we identify the path with the largest count by approximating the count over the input space $\mathbb I$ along each path. However, there will always be uncertainty in this path identification process due to approximation error. 

To capture the uncertainty from our approximate counts, we add a correction factor $\sqrt{\frac{\log(t)}{T_{\pi}}}$ shown in Algorithm~\ref{alg:noisycumulativecountoracle} to also explore alternative paths a small number of times, borrowed from the multi-armed bandit literature~\cite{multiarmed}. $C_{\pi}$ denotes our latest count of inputs that reach the target along path $\pi$ and $T_{\pi}$  denotes number of times path $\pi$ has been selected prior to the $t$-th oracle query. This correction factor conceptually balances uncertainty because as the algorithm acquires more certainty about a path $\pi$ by selecting it more, thereby increasing $T_{\pi}$, the correction factor gradually decreases as the term is inversely proportional to $T_{\pi}$. We keep track of the most recent count information per path through a cache data structure called the \Graph.

\begin{figure}[t!]
\vspace{-0.5cm}
\begin{algorithm}[H]
\footnotesize
\caption{\small Noisy Counting Oracle.} 
\label{alg:noisycumulativecountoracle} 
\lstset{basicstyle=\ttfamily\footnotesize, breaklines=true}
\begin{tabular}{|lp{2.6in}|}\hline
\textbf{Input}:
    & \textit{$I_L$} $\leftarrow$ Left Input Region\\
    & \textit{$I_R$} $\leftarrow$ Right Input Region\\
\hline
\end{tabular}
\begin{algorithmic}[1]

\State $\textsf{Counts = HashMap()}$ 
\For{$\pi \in \Pi^T$ }  
    \If{$\pi \not\in \text{\Graph}$}
        \State $C_{\pi}$ = ApproxCount($\mathbb I, \pi$) \Comment{\textcolor{purple}{Initialize \Graph \ with count over $\mathbb I$}}
        \State Insert $(C_{\pi}, 1)$ into \Graph
    \EndIf
    \State Lookup $(C_{\pi}, T_{\pi})$ in \Graph
    \State $\textsf{Counts[$\pi$]} = C_{\pi} + \sqrt{\frac{\log(t)}{T_{\pi}}} $ \Comment{\textcolor{purple}{Uncertainty term on $t$-th oracle query}}
    
\EndFor
\State $\pi = \argmax_{i \in \textsf{Counts}} (\textsf{Counts[i]})$ \Comment{\textcolor{purple}{Select path with largest count}}

\State

\State $C(I_L), C(I_R)$ = ApproxCount($I_L, \pi$), ApproxCount($I_R, \pi$)

\State   $C_{\pi} = \max(C(I_L), C(I_R))$ \Comment{\textcolor{purple}{Update count based on latest information}}
\State  Update \Graph \ entry for $\pi$ with $(C_{\pi}, T_{\pi} + 1)$

\If{$C(I_L) \geq C(I_R)$} \Comment{\textcolor{purple}{Send back answer to Algorithm~\ref{alg:noisyfuzzer}}}
    \State \textbf{return} 1 
\Else
    \State \textbf{return} 0
\EndIf
\State 
\Procedure{ApproxCount}{$I$, $\pi$} \Comment{\textcolor{purple}{Approximate $C_{\pi}(I)$}} %
    \State $\textsf{Inputs = Select k uniformly random inputs from I}$
    \State $ \textsf{ratios = MonteCarloExecutions(Program, $\pi$, Inputs)}$ \Comment{\textcolor{purple}{Algorithm~\ref{alg:mcex} }}
    \State $\textsf{\textbf{return} } |I| * \min(\textsf{ratios})$ \Comment{\textcolor{purple}{Equation~\ref{equation:upperbound} }}
\EndProcedure

\end{algorithmic}
\end{algorithm}
\vspace{-0.4cm}
\end{figure}

\section{Implementation}
\label{section:implementation}
\noindent\textbf{Toolchain.}
We implement algorithms \ref{alg:noisyfuzzer}, \ref{alg:mcex}, \ref{alg:noisycumulativecountoracle} in C. We use LLVM~\cite{llvm2004} instrumentation and signal handlers to handle the branch and exception logic, respectively in Algorithm~\ref{alg:mcex}. We also incorporate the fork-server optimization used in state-of-the-art fuzzers~\cite{aflgo, parmesan, afl, angora}. As described in Section \ref{section:methodologycounting}, the error of the sample mean and variance drops exponentially fast in the number of program executions $\sim \frac{1}{e^{k}}$, so we set $k=5$ in Algorithm \ref{alg:noisycumulativecountoracle} such that $p=0.01$ in Algorithm \ref{alg:noisyfuzzer} because $\frac{1}{e^5} \leq 0.01$. To reduce storage overheads, we implement the \Graph \ as a trie which avoids duplication when paths share edges. Moreover, we do not track $k$ branch distances per branch in Algorithm \ref{alg:mcex}, but rather compute the sample mean and variance in a streaming setting~\cite{streamingmean}, so that we only store a constant number of values for any number of program executions $k$. Such techniques contribute to our minimal performance overheads in Section \ref{sec:evaluationrq3}.

\vspace{0.1 cm}
\noindent\textbf{Reducing Loop Overheads.} Real-world programs use loops causing the same branch to be visited many times during a program execution. If a single branch is visited a million times per execution, a naive implementation of Algorithm \ref{alg:mcex} will store a million branch distances for this branch per execution. Instead, we share information across multiple visits to a branch to reduce loop storage overheads.  Specifically, in a single \Execution, if a branch is visited multiple times, the branch distance at each visit contributes to the (streaming) mean and variance of the branch. We also enforce control-flows at runtime across multiple visits to a branch by attaching count information to each branch. In addition to the techniques mentioned earlier, these techniques better help us scale to large real-world programs and contribute to our minimal overheads in Section~\ref{sec:evaluationrq3}.

\vspace{0.1 cm}
\noindent\textbf{Assigning A Total Order.}
In Section \ref{section:methodologymodel}, we use the lexicographic total order (i.e., flatten first byte, then second byte, and so forth) to unambiguously split the input space. Although noisy binary search is agnostic to the underlying total order, using lexicographic order in real-world programs assumes that any region of the input space is equally likely to change the counts of inputs that reach the target (i.e., all bytes equally contribute to program behavior). However, for many real-world programs, this assumption does not hold as experimental evidence shows that not all bytes equally contribute to program behaviors ~\cite{neutaint, neuzz, redqueen, greyone, hotbytes}. 

Therefore, instead of assigning a total order based on lexicographic order, we assign an order based on the observed program executions in the noisy counting oracle. 
Specifically, starting with the set of all byte indices, the algorithm partitions the set into two disjoint subsets of equal size, and for each subset, performs \Execution\  on inputs generated by perturbing byte values whose index belongs to the subset. If the program executions change the approximate count, the algorithm recursively repeats the prior step on the subset. Otherwise, the subset is ignored. The algorithm repeats this process until the only sets that remain are sets with a single byte index. We then assign a total order by prioritizing byte indices from these remaining sets ranked by how much each byte index increases the approximate count. We experimentally demonstrate the effectiveness of this approach in Appendix \ref{appendix:totalorder}. 

\vspace{0.1 cm}
\noindent\textbf{Preprocessing.} Existing work in directed greybox fuzzing~\cite{aflgo, hawkeye, beacon} pre-computes information about a program (e.g., static analysis information or CFG distance) to better guide the directed greybox fuzzer. In our setting, we need to pre-compute the set of all paths that reach the target, a task where algorithms require prohibitively expensive runtimes over large real-world CFGs~\cite{simplepaths, graphalgo}. Moreover, algorithms that generate a subset of paths ~\cite{simplepaths} generally do not produce paths with repeated edges, and since loops are a common construct in real-world programs, the set of generated paths is unlikely to be realizable in real program executions. 

We instead use the initial seed corpus to bootstrap a set of paths. Specifically, we generate a set of paths that reach the target by executing the program on a seed close to the target while randomly inverting the direction of branches along the corresponding execution path, keeping paths based on the program executions after the branches were inverted (e.g., reach the target). Consequently, we use this seed's length to set the input region size. In Appendix~\ref{appendix:preprocessing}, we measure our preprocessing time, comparing it to that of directed greybox fuzzers to show that our preprocessing times are similar. We plan to explore better path generation strategies in future work, potentially using ideas from the symbolic execution literature~\cite{klee, qsym, symcc, savior, Godefroid2008AutomatedWF}. 

\vspace{0.1 cm}

\noindent\textbf{Randomly Generating Inputs.} 
We represent the input region as a $d$-dimensional hyperrectangle encoded as $d$ intervals, where each interval represents upper and lower bounds on input values per dimension. Used in Algorithm~\ref{alg:noisycumulativecountoracle}, we select $k$ inputs uniformly at random from the hyperrectangle by generating $d$ integers independently and uniformly at random from each interval, repeating this process $k$ times for $k$ inputs of length $d$. If the initial seed belongs to a given hyperrectangle (i.e., the seed’s byte values are within the $d$ intervals), we include it as part of the $k$ inputs to better utilize initial seed corpus information when applicable.

Note that we do not keep track of a seed corpus. Instead, we keep track of a list of groups as shown in Algorithm~\ref{alg:noisyfuzzer}, where each group corresponds to a tuple: (hyperrectangle, weight) and splitting an input region corresponds to adjusting the hyperrectangle’s per-dimension intervals. To mitigate potential error in the input region weights if the selected path changes during the oracle queries, we also keep track of the groups per path, which does not introduce significant storage overhead as shown in our performance overheads in Section~\ref{sec:evaluationrq3} since our algorithm uses logarithmic number of groups in expectation with respect to the size of $\mathbb I$.

\section{Evaluation}
\label{sec:evaluation}
Our evaluation seeks to answer the following research questions.

\begin{enumerate}
    \item \textbf{Comparison against directed greybox fuzzers:} How does \ToolName{} compare to state-of-the-art directed greybox fuzzers?
    \item \textbf{Bug Finding:} Can \ToolName{} find new real-world bugs?
    \item \textbf{Performance Overhead:} What is the performance overhead of \ToolName{}?
    \item \textbf{Design Choices:} Are \ToolName{}'s design choices justified?
\end{enumerate}

\noindent\textbf{Compute Infrastructure. }
Unless otherwise noted, we ran all experiments on a Ubuntu 18.04 workstation with a Ryzen Threadripper 2970WX 24-Core CPU and 128 GB RAM.

\subsection{RQ1: Fuzzers Comparison}
\label{sec:evaluationrq1}

\noindent\textbf{Tested Benchmarks. }
To avoid any potential bias while creating our own CVE benchmark in terms of bug class or program type, we use the publicly available Magma benchmark~\cite{magma}, which was specifically curated from a diverse set of CVEs. We also evaluate on a subset of the Fuzzer Test Suite benchmark~\cite{fuzzerTestSuite} covered by prior work~\cite{hawkeye, parmesan} to enable fair comparison.

\vspace{0.1 cm}
\noindent\textbf{Baseline Fuzzers. }
Following prior works in directed greybox fuzzing \cite{beacon, hawkeye, cafl, fuzzguard}, we primarily compare \ToolName{} against other directed greybox fuzzers like \texttt{AFLGo}~\cite{aflgo}. Other directed greybox fuzzers are either not available in any form (source or binary)~\cite{fuzzguard, hawkeye, cafl} or have not made their source code public yet ~\cite{beacon} and cannot support our benchmarks (i.e., Magma and Fuzzer Test Suite) without significant modifications. We also reached out to the authors of several of these fuzzers and confirmed that their code is not available for a release at the time of this writing, but they are working on releasing their code soon. Therefore, we could not compare against them on our benchmarks (i.e., Magma and Fuzzer Test Suite).

To compare against alternative designs for directed greybox fuzzing other than \texttt{AFLGo}, we also evaluate \ToolName{} against \texttt{ParmeSan}~\cite{parmesan} which supports a directed greybox fuzzer mode. We contacted the authors of \texttt{ParmeSan} and followed their advice to set it up. Furthermore, as \texttt{ParmeSan} and \texttt{AFLGo} build upon two significantly different regular (i.e., undirected) fuzzers: \texttt{Angora}~\cite{angora} and \texttt{AFL}~\cite{afl}, respectively, we also include the results of the underlying fuzzer implementations to show the improvement a directed greybox fuzzer has over its undirected counterpart in Appendix Tables~\ref{tab:magma4} and~\ref{tab:fts4}.

\vspace{0.1 cm}
\noindent\textbf{Experimental Setup. }
We follow the experimental setup based on prior work \cite{aflgo, cafl, parmesan, beacon, hawkeye, fuzzguard}. We assign each fuzzer a single core and keep 20\% of the cores unused to minimize interference. We configure each directed greybox fuzzer to use the default seeds and targets provided by the benchmarks. 
To avoid potential unfairness or bias in the results arising from how different fuzzing implementations deal with multiple targets, we give fuzzers one target per run to enable a fair comparison in line with prior work~\cite{cafl, parmesan}.
We measure the time it takes to trigger the bug target (for Magma) or reach the target (for Fuzzer Test Suite) with a 6 hour timeout. 

We pick 6 hours because it is the arithmetic mean of the times used by \texttt{Hawkeye}~\cite{hawkeye} and \texttt{AFLGo}~\cite{aflgo} evaluations. Since each fuzzer includes some amount of preprocessing (e.g., distance computations), we also separately measure this time in Table \ref{tab:preprocessingmagma} in Appendix \ref{appendix:preprocessing}. We run with 20 independent trials, using arithmetic mean when reporting results. We note that our Fuzzer Test Suite experiments were performed on a workstation running Ubuntu 18.04 using an Xeon E5-2640 24-Core CPU with 128 GB RAM.

\begin{table}[!]
\caption{\label{tab:magma}\small\textbf{\small Mean time to trigger Magma bugs for each tested fuzzer over 20 trials. We only include the bugs that were triggered within 6 hours for space constraints. (x) refers to the speedup of \ToolName{} relative to the tested fuzzer. (p) refers to the p-value from the Mann-Whitney U test. Since \texttt{ParmeSan} crashed on \texttt{php}, we write N/A for it. $T.O^\star$ indicates 6 hour timeout. We highlight bugs only triggered by \ToolName{} in \colorbox[HTML]{ECF4FF}{blue}. }}
    \vspace{-0.2cm}
\footnotesize
\centering
\setlength{\tabcolsep}{2.5pt}
\begin{tabular}{|l|c|ccr|ccr|}
\hline
\rowcolor[HTML]{EFEFEF}
 & \multicolumn{1}{c|}{\textbf{\ToolName{}}} & \multicolumn{3}{c|}{\textbf{\texttt{AFLGo}}} & \multicolumn{3}{c|}{\textbf{\texttt{ParmeSan}}}  \\
\rowcolor[HTML]{EFEFEF} \multirow{-2}{*}{\textbf{Bug ID}}
 & \multicolumn{1}{c|}{\textbf{Time}} & \multicolumn{1}{c}{\textbf{Time}}& \multicolumn{1}{c}{\textbf{$(x)$}} &\multicolumn{1}{c|}{\textbf{$(p)$}} & \multicolumn{1}{c}{\textbf{Time}} &  \multicolumn{1}{c}{\textbf{$(x)$}} & \multicolumn{1}{c|}{\textbf{$(p)$}}  \\
\hline
\texttt{PDF010} & 3m15s & 4h02m15s & 74x      & $<$0.01 & $T.O^\star$    & $>$110x  & $<$0.01 \\
\texttt{PDF016} & 3m23s & 51m43s   & 15x      & $<$0.01 & 7m10s    & 2x       & $<$0.01 \\
\texttt{PHP004} & 1m04s & 4m09s    & 3x       & $<$0.01 & N/A      & N/A      & N/A      \\
\texttt{PHP009} & 1m07s & 17m08s   & 15x      & $<$0.01 & N/A      & N/A      & N/A      \\
\texttt{PHP011} & 1m01s & 15m24s   & 15x      & $<$0.01 & N/A      & N/A      & N/A      \\
\texttt{PNG003} & 15s   & 15s      & 1x       & 0.25    & 1m38s    & 6x       & $<$0.01 \\
\texttt{PNG006} & 1m36s & $T.O^\star$    & $>$225x  & $<$0.01 & 2m03s    & 1x       & $<$0.01 \\
\texttt{SSL002} & 1m44s & 5m58s    & 3x       & $<$0.01 & 32m27s   & 18x      & $<$0.01 \\
\texttt{SSL003} & 1m39s & 4m30s    & 2x       & $<$0.01 & 16m27s   & 9x       & $<$0.01 \\
\texttt{SSL009} & 4m59s & $T.O^\star$    & $>$72x   & $<$0.01 & 4h51m19s & 58x      & $<$0.01 \\
\texttt{TIF005} & 9m33s & $T.O^\star$    & $>$37x   & $<$0.01 & 3h48m49s & 23x      & $<$0.01 \\
\texttt{TIF006} & 9m36s & $T.O^\star$    & $>$37x   & $<$0.01 & 4h03m29s & 25x      & $<$0.01 \\
\texttt{TIF007} & 8m18s & 1h39m40s & 12x      & $<$0.01 & 56m40s   & 6x       & $<$0.01 \\
\texttt{TIF012} & 9m59s & 2h46m00s & 16x      & $<$0.01 & 3h52m50s & 23x      & $<$0.01 \\
\texttt{TIF014} & 1m36s & 5h49m19s & 218x     & $<$0.01 & $T.O^\star$    & $>$225x  & $<$0.01 \\
\texttt{XML017} & 16s   & 1m09s    & 4x       & $<$0.01 & 23m15s   & 87x      & $<$0.01 \\
\rowcolor[HTML]{ECF4FF}
\texttt{PDF003} & 1m39s & $T.O^\star$    & $>$218x  & $<$0.01 & $T.O^\star$    & $>$218x  & $<$0.01 \\
\rowcolor[HTML]{ECF4FF}
\texttt{PDF008} & 3m21s & $T.O^\star$    & $>$107x  & $<$0.01 & $T.O^\star$    & $>$107x  & $<$0.01 \\
\rowcolor[HTML]{ECF4FF}
\texttt{PDF011} & 1m41s & $T.O^\star$    & $>$213x  & $<$0.01 & $T.O^\star$    & $>$213x  & $<$0.01 \\
\rowcolor[HTML]{ECF4FF}
\texttt{PDF018} & 1m43s & $T.O^\star$    & $>$209x  & $<$0.01 & $T.O^\star$    & $>$209x  & $<$0.01 \\
\rowcolor[HTML]{ECF4FF}
\texttt{PDF019} & 1m37s & $T.O^\star$    & $>$216x  & $<$0.01 & $T.O^\star$    & $>$216x  & $<$0.01 \\
\rowcolor[HTML]{ECF4FF}
\texttt{PNG001} & 3m17s & $T.O^\star$    & $>$109x  & $<$0.01 & $T.O^\star$    & $>$109x  & $<$0.01 \\
\rowcolor[HTML]{ECF4FF}
\texttt{PNG007} & 3m21s & $T.O^\star$    & $>$107x  & $<$0.01 & $T.O^\star$    & $>$107x  & $<$0.01 \\
\rowcolor[HTML]{ECF4FF}
\texttt{SSL020} & 9m16s & $T.O^\star$    & $>$38x   & $<$0.01 & $T.O^\star$    & $>$38x   & $<$0.01 \\
\rowcolor[HTML]{ECF4FF}
\texttt{TIF001} & 9m43s & $T.O^\star$    & $>$37x   & $<$0.01 & $T.O^\star$    & $>$37x   & $<$0.01 \\
\rowcolor[HTML]{ECF4FF}
\texttt{TIF002} & 9m58s & $T.O^\star$    & $>$36x   & $<$0.01 & $T.O^\star$    & $>$36x   & $<$0.01 \\
\rowcolor[HTML]{ECF4FF}
\texttt{TIF009} & 9m49s & $T.O^\star$    & $>$36x   & $<$0.01 & $T.O^\star$    & $>$36x   & $<$0.01 \\
\rowcolor[HTML]{ECF4FF}
\texttt{XML009} & 13s   & $T.O^\star$    & $>$1661x & $<$0.01 & $T.O^\star$    & $>$1661x & $<$0.01\\
\hline
\multicolumn{2}{|c|}{\textbf{Mean speedup}}   & & 134x & & & 144x & \\
\multicolumn{2}{|c|}{\textbf{Median speedup}} & & 38x  & & & 39x  & \\

\hline
\end{tabular}
\vspace{-2mm}
\end{table}

\vspace{0.1 cm}
\noindent\textbf{Magma Results. } Table \ref{tab:magma} summarizes the results as well as the result from applying the Mann-Whitney U test between \ToolName{} and the tested directed greybox fuzzers. We note that although we evaluated over the entire benchmark, not all bugs were triggered, and therefore, for space constraints, we only list the bugs triggered within the time budget in Table~\ref{tab:magma} following prior work \cite{seed_select}. 

\ToolName{} finds bugs 134x faster in arithmetic mean and 38x faster in median compared to the next best fuzzer \texttt{AFLGo}. Moreover, \ToolName{}'s improvement is statistically significant with a significance level of 0.05 for all bugs except \texttt{PNG003}. \ToolName{} was also able to find 28 bugs in total, 16 more than the next-best fuzzer \texttt{AFLGo}, which found only 12 bugs within the time budget. We note that since \ToolName{} does not generate inputs of different length, we also ran this experiment with variants of \texttt{AFLGo} and \texttt{ParmeSan} that do not change the input length.  We found the results to be nearly identical (mean speedup changed by $2\%$), so we did not insert the full table for space constraints. Overall, our results show the promise of using noisy binary search and approximate counting for directed greybox fuzzing.

\begin{figure}[t!]
\lstset{basicstyle=\footnotesize\ttfamily, breaklines=true,numbers=left}%
\begin{lstlisting}[]
void png_check_chunk_length() {
  // set based on input file
  u32 width, height, colortype;
  
  /* constraints from libpng_read_fuzzer.cc */
  if (width < UINT_31_MAX) {
    if (width*height < 10^8) {
      u32 channels;
      switch(colortype) {
        case PALETTE: channels = 1; break
        case GRAY: channels = 2; break;
        case RGB: channels = 3; break;
        case ALPHA: channels = 4; break;
      }
      
      u32 row_factor = width * channels + 1;
      if (row_factor == 0) {
        // divide-by-zero bug target
      }
    }
  }
}
\end{lstlisting}
\vspace{-0.4cm}
\caption{\label{fig:achilles}\small Simplified code of Magma PNG001 (CVE-2018-13785).
}
\end{figure}

\noindent\textbf{Case Study. } We highlight a particular bug \texttt{PNG001} in Figure ~\ref{fig:achilles} found only by \ToolName{}. 
This bug is guarded by constraints and only a single input value \texttt{width=0x55555555} will cause a divide by zero when \texttt{row\_factor} overflows. We hypothesize that \texttt{AFLGo} did not trigger this bug in the time budget because the chance of producing this specific input value through mutations is small and fuzzer heuristics such as setting values to \texttt{MAX\_INT} also fail. In addition, we hypothesize that \texttt{ParmeSan} did not trigger this bug in the time budget because although it uses gradient descent and taint tracking to narrow down the input space, it cannot effectively reason about nested constraints (Lines 6 and 7). In contrast, \ToolName{} was able to successfully find this input value through noisy binary search. Moreover, upon manual source code analysis, we found this bug can only be triggered along an execution path that sets \texttt{channels=3}, showing that \ToolName{} was able to successfully reason across multiple execution paths.

\begin{table}[!]
\caption{ \label{tab:fts} \small\textbf{\small 
 Mean time to reach Fuzzer Test Suite targets for each tested fuzzer over 20 trials.
(x) refers to the speedup of \ToolName{} relative to the tested fuzzer. (p) refers to the p-value from the Mann-Whitney U test. $T.O^\star$ indicates 6 hour timeout.
We highlight targets only reached by \ToolName{} in \colorbox[HTML]{ECF4FF}{blue}.}}
    \vspace{-0.2cm}
\footnotesize
\centering
\setlength{\tabcolsep}{2.5pt}
\begin{tabular}{|l|c|ccr|ccr|}
\hline
\rowcolor[HTML]{EFEFEF}
 & \multicolumn{1}{c|}{\textbf{\ToolName{}}} & \multicolumn{3}{c|}{\textbf{\texttt{AFLGo}}} & \multicolumn{3}{c|}{\textbf{\texttt{ParmeSan}}}  \\
\rowcolor[HTML]{EFEFEF} \multirow{-2}{*}{\textbf{Bug ID}}
 & \multicolumn{1}{c|}{\textbf{Time}} & \multicolumn{1}{c}{\textbf{Time}}& \multicolumn{1}{c}{\textbf{$(x)$}} &\multicolumn{1}{c|}{\textbf{$(p)$}} & \multicolumn{1}{c}{\textbf{Time}} &  \multicolumn{1}{c}{\textbf{$(x)$}} & \multicolumn{1}{c|}{\textbf{$(p)$}}  \\
\hline
\texttt{ttgload.c:1710}  & 1s    & 1s    & 1x      & 0.07    & 1s    & 1x      & 0.07      \\
\texttt{ttinterp.c:2186} & 9m57s & $T.O^\star$ & $>$36x  & $<$0.01 & 20m   & 2x      & $<$0.01   \\
\texttt{cf2intrp.c:361}  & 58s   & 23m   & 23x     & $<$0.01 & $T.O^\star$ & $>$372x & $<$0.01   \\
\texttt{jdmarker.c:659}  & 32s   & 1h07m & 125x    & $<$0.01 & 5m    & 9x      & $<$0.01   \\
\texttt{pngrutil.c:139}  & 1s    & 1s    & 1x      & 0.07    & 1s    & 1x      & 0.07      \\
\texttt{pngrutil.c:3182} & 28s   & 2m30s & 5x      & $<$0.01 & 1m    & 2x      & $<$0.01   \\
\texttt{pngread.c:738}   & 1s    & 1s    & 1x      & 0.07    & 1s    & 1x      & 0.07      \\
\rowcolor[HTML]{ECF4FF}
\texttt{pngrutil.c:1393} & 51s   & $T.O^\star$ & $>$423x & $<$0.01 & $T.O^\star$ & $>$423x & $<$0.01   \\
\hline
\multicolumn{2}{|c|}{\textbf{Mean speedup}}   & & 77x & & & 102x & \\
\multicolumn{2}{|c|}{\textbf{Median speedup}} & & 15x  & & & 2x   & \\
\hline
\end{tabular}
\vspace{-0.3cm}
\end{table}

\vspace{0.1 cm}
\noindent\textbf{Fuzzer Test Suite Results. }
Table \ref{tab:fts} summarizes the results. \ToolName{} reaches targets 102x faster in arithmetic mean and 2x faster in median than \texttt{ParmeSan} and 77x faster in arithmetic mean and 15x faster in median than \texttt{AFLGo}, with statistical significance on all targets that were not reached within a few seconds. Moreover, \ToolName{} reaches 2 more targets compared to either \texttt{ParmeSan} or \texttt{AFLGo}. While cross-comparisons between papers is challenging due to stochasticity in fuzzers and hardware, our results are similar to prior work \cite{hawkeye, parmesan}, giving us confidence in our experimental setup of the tested fuzzers.

\vspace{0.3cm}
\begin{longfbox}
\textbf{Result 1:} 
Over the Magma benchmark, \ToolName{} finds bugs 134x faster in arithmetic mean and 38x faster in median compared to the next best fuzzer \texttt{AFLGo}. It also finds 28 bugs in total, 16 more than the next-best fuzzer \texttt{AFLGo}. 
\end{longfbox}

\subsection{RQ2: Bug Finding}
\label{sec:evaluationrq2}

\begin{table}[!]
\caption{ \label{tab:studied_programs}\small\textbf{Tested programs in bug finding experiments.}}
    \vspace{-0.2cm}
\centering
\footnotesize
\begin{tabular}{llr}
\toprule
\multicolumn{1}{c}{\textbf{Library}}    & \multicolumn{1}{c}{\textbf{Program}} & \multicolumn{1}{c}{\textbf{Version}} \\
\midrule
\texttt{libpng} & \texttt{libpng\_read\_fuzzer}     &  Commit a37d483...  \\
\texttt{poppler} & \texttt{pdf\_fuzzer}    & Commit 1d23101... \\
\texttt{binutils} & \texttt{nm -C}    & 2.36  \\
\texttt{binutils} & \texttt{objdump -xD}        & 2.36 \\
\texttt{openssl} & \texttt{x509}    & Commit 3bd5319...  \\
\texttt{libxml2} & \texttt{xmllint}    & Commit 07920b4... \\
\toprule
\end{tabular}
\vspace{-0.2cm}
\end{table}

For our bug finding experiments, we evaluate over programs based on prior work~\cite{fairfuzz, neuzz, ecofuzz, tfuzz, mopt, jung:fuzzification, tortoise} and Magma listed in Table~\ref{tab:studied_programs}. To find the targets for directed fuzzing, we re-use an idea from prior work~\cite{parmesan, hawkeye} and use \texttt{Undefined Behavior Sanitizer}~\cite{ubsan} to identify bug targets. This tool often reports a large number of bug targets, and if all are set as targets, the fuzzer effectively becomes a coverage-guided fuzzer instead of being directed. Instead, we randomly pick one target per function and run each fuzzer with these same targets over a 24 hour run. We start each fuzzer with the initial Magma corpus and a small set of valid ELF files. We report the total number of bugs found, repeating this experiment 10 times to minimize variability. 

\begin{table}[h!]
\caption{ \label{tab:bugs}\small\textbf{ Categorization of new bugs found by each fuzzer.}}
    \vspace{-0.2cm}
    \centering
    \footnotesize
\begin{tabular}{lrrr}
\toprule
\multicolumn{1}{c}{\textbf{Bug Type}}    & \multicolumn{1}{c}{\textbf{\texttt{ParmeSan}}} & \multicolumn{1}{c}{\textbf{\texttt{AFLGo}}} & \multicolumn{1}{c}{\textbf{\ToolName{}}} \\
\midrule
divide-by-zero     &  0    & 0    & 1    \\
denial-of-service     &  3    & 4    & 6    \\
stack/heap overflow       & 10     & 8    & 13     \\
integer overflow  & 21     & 17    & 29   \\
\midrule
Total             & 34     & 29    & 49  \\
\bottomrule
\end{tabular}
\end{table}

In our 24 hour runs, we found previously-unknown real-world bugs in \texttt{binutils}, \texttt{libxml2}, and \texttt{libpng}. Table~\ref{tab:bugs} summarizes the results in terms of bug type. Since counting the number of crashing inputs may inflate the bug count, we take the following approach to better compute the bug count based on prior work~\cite{redqueen, angora, neuzz}. We first use \texttt{AFL-CMin} to filter out duplicate crashing inputs, followed by another deduplication procedure based on unique stack traces. From this reduced set of inputs, we manually review the stack traces and corresponding source code to further deduplicate these inputs. We responsibly disclosed these bugs to the developers and all bugs were confirmed, most of which have been fixed in the latest versions of the programs. Our results show that \ToolName{} finds 15 more bugs than the next best fuzzer \texttt{ParmeSan}.

\begin{longfbox}
\textbf{Result 2:} 
\ToolName{} finds 49 previously-unknown real world bugs, 15 more than the next best fuzzer \texttt{ParmeSan}.  
\end{longfbox}

\subsection{RQ3: Performance Overhead}
\label{sec:evaluationrq3}
\begin{table}[!]
\caption{\label{tab:runtimeoverheads}\small\textbf{ \Execution \ Overheads relative to native (uninstrumented) and fuzzer-instrumented execution over Magma.  }}
    \vspace{-0.2cm}
    \centering
    \footnotesize
\setlength{\tabcolsep}{2.6pt}
\begin{tabular}{c|rr|rr}
   \toprule
    \multicolumn{1}{c}{\multirow{2}{*}{\textbf{Library}}} & \multicolumn{2}{c}{\textbf{\texttt{\ToolName{} vs Native}}} & \multicolumn{2}{c}{\textbf{\texttt{\ToolName{} vs Fuzzer (AFLGo)}}}  \\ \cline{2-5} 
\multicolumn{1}{c}{}& \multicolumn{1}{c}{\textbf {Runtime}} & \textbf{Memory}     & \multicolumn{1}{c}{\textbf {Runtime}} & \textbf{Memory}  \\    
\midrule
\texttt{libpng}     & 94\%  &  30\%     & 26\% & 4\%    \\
\texttt{libtiff}    & 78\%    &  2\%    & 16\%  &  1\% \\
\texttt{libxml2}    & 135\%  &    37\%  & 38\%  & 8\% \\
\texttt{openssl}    & 117\%   &  15\%    & 42\%  & 6\%  \\
\texttt{php}        & 86\%   &   9\%   & 29\% & 4\%   \\
\texttt{poppler}    & 87\%   &  10\%   & 25\%  & 3\%  \\
\texttt{sqlite3}    & 136\%  &   7\%  & 34\%  & 4\% \\
\midrule
\multicolumn{1}{c|}{Arithmetic mean} & 105\% & 16\% &  30\% & 4\% \\
\multicolumn{1}{c|}{Median } & 94\% & 10\% & 29\% & 4\%\\
\bottomrule
\end{tabular}
\vspace{-0.3cm}
\end{table}

\begin{table}[!]
     \caption{\label{tab:fuzzmapsize}\small\textbf{ \ToolName{}'s data structures size in MBs over Magma benchmark.}}
         \vspace{-0.2cm}
    \centering
    \footnotesize
\begin{tabular}{cc}
\toprule
\multicolumn{1}{c}{\textbf{Library}}    & \multicolumn{1}{c}{\textbf{Data Structures Size (MBs)}}  \\
\midrule
\texttt{libpng}     & 12.1    \\
\texttt{libtiff}    & 21.8    \\
\texttt{libxml2}    & 1.6     \\
\texttt{openssl}    & 59.7    \\
\texttt{php}        & 1.6     \\
\texttt{poppler}    & 28.1    \\
\texttt{sqlite3}    & 20.8    \\
\midrule
\multicolumn{1}{c|}{Arithmetic mean  } & 20.9 \\
\multicolumn{1}{c|}{Median } & 20.8 \\
\bottomrule
\end{tabular}
\vspace{-0.5cm}
\end{table}

Since instrumented target program executions dominate the fuzzing overhead~\cite{fullspeedfuzzing}, we evaluate the performance overhead of \Execution \ relative to native (uninstrumented) execution as well as a fuzzer-instrumented execution that tracks edge coverage and distance (i.e., \texttt{AFLGo}). We run the Magma programs over the initial seed corpus inputs and take the arithmetic mean of the results from 10 independent trials. In addition, we measure the total memory footprint of \ToolName{}'s data structures (e.g., \Graph \ and weight groups in Algorithm \ref{alg:noisyfuzzer}) by re-running our Magma evaluation and tracking the total memory consumed in MBs, reporting the arithmetic mean over 10 independent trials. 

Table \ref{tab:runtimeoverheads} summarize the performance overheads of \Execution. \Execution \ adds runtime overheads of 105\% in arithmetic mean and 94\% in median as well as memory overheads of 16\% in arithmetic mean and 10\% in median relative to native execution. Relative to a fuzzer-instrumented execution, the overheads are smaller: runtime overheads of 30\% in arithmetic mean and 29\% in median as well as memory overheads of 4\% in arithmetic mean and 4\% in median. We attribute the additional memory and runtime overheads to computing the (streaming) mean and variance for each branch, which requires additional memory as well as floating point arithmetic. 

We also summarize the memory footprint: 20.9 MBs in arithmetic mean and 20.8 MBs in median ($< 1$ GB) with full details in Table \ref{tab:fuzzmapsize}. These results show that the data structures do not consume large amounts of memory. Note that \ToolName{}, a prototype, still consistently outperforms other fuzzers despite this overhead, showing the promise of our technique. Nonetheless, we believe there are still ways to further cut down our prototype's overhead.

\vspace{0.3cm}
\begin{longfbox}
\textbf{Result 3:} 
\ToolName{} adds 30\% runtime and 4\% memory overhead in arithmetic mean relative to a fuzzer's instrumentation and 105\% runtime and 16\% memory overhead in arithmetic mean relative to native execution. In addition, \ToolName{} data structures consume $< 1$ GB of memory. 

\end{longfbox}
\vspace{-0.4cm}
\subsection{RQ4: Design Choices}
\label{sec:evaluationrq4}

We conduct experiments to measure the effect of three design choices:  (i) Chebyshev's inequality for uniconstraint counts, (ii) using the minimum uniconstraint count, and (iii) path selection.  

For each design choice experiment, we run \ToolName{} on a representative subset from the Magma benchmark, repeated 10 times. To form a representative subset, we pick 3 bugs randomly from three categories: bugs found within 60 seconds, bugs found more than 120 seconds, and bugs found between these times. Our subset includes at least one bug from each library in Magma. Moreover, it includes bugs that only \ToolName{} triggers as well as other tested fuzzers trigger.  We describe each design choice experiment in more detail below.

\subsubsection{Chebyshev's Inequality for Uniconstraint Counts}
In this experiment, we compare our Chebyshev-based technique to compute probabilistic upper bounds on $r$ (i.e., see Equation~\ref{equation:uniconstraintcount} in Section~\ref{section:methodologycounting})  against alternate techniques when $r=0$ (i.e., zero uniconstraint counts). Specifically, we compare against the Rule-of-3 and Good-Turing techniques from the Natural Language Processing and Biostatistics literature~\cite{smoothing, goodturingbiostatistics}, which have also been used in prior work in fuzzing~\cite{ruleofthree, goodturingmarcel}. In contrast to our probabilistic upper bounds which use mean and variance information, these methods upper bound $r$ by computing $r=\frac{3}{N}$ (Rule-of-3) or the smallest non-zero $r$ across all branches (Good-Turing) via $r=min(\{r_{E_1}, r_{E_2}, ..., r_{E_T} \text{ such that } r_{E_i} \neq 0\})$. 

\begin{table}[!]
\caption{\label{tab:aba_smoothing}\small\textbf{ Mean time to trigger the bug across various techniques to approximate uniconstraint counts over 10 trials.}}
\vspace{-0.2cm}
\centering
\footnotesize
\begin{tabular}{lccc}
\toprule
\multicolumn{1}{l}{\textbf{Bug ID}}   & \multicolumn{1}{c}{\textbf{\ToolName{}}}   & \multicolumn{1}{c}{\textbf{Rule-Of-3}}& \multicolumn{1}{c}{\textbf{Good-Turing}} \\
\midrule
\texttt{XML009}  & 13s  & $T.O^\star$ & $T.O^\star$     \\
\texttt{PNG003}  & 15s  & $T.O^\star$ & $T.O^\star$     \\
\texttt{XML017} & 16s   & 2m17s & 1m54s     \\
\texttt{PHP004}  & 1m04s  & 15m8s & 10m40s  \\
\texttt{PDF011} & 1m41s   & $T.O^\star$ & $T.O^\star$   \\
\texttt{PHP009}  & 1m07s  & $T.O^\star$ & $T.O^\star$   \\
\texttt{SSL020}  & 9m16s  & $T.O^\star$ & $T.O^\star$   \\
\texttt{TIF009}   & 9m49s  & $T.O^\star$ & $T.O^\star$  \\
\texttt{PDF008}   & 3m21s  & $T.O^\star$ & $T.O^\star$  \\
\midrule
\multicolumn{2}{c|}{Arithmetic mean speedup } & 427x & 426x  \\
\multicolumn{2}{c|}{Median speedup}           & 107x & 107x  \\
        \bottomrule
\end{tabular}
\vspace{-0.2cm}
\end{table}
Table \ref{tab:aba_smoothing} summarizes the results. \ToolName{} improves upon the next-best technique Good-Turing by 426x in arithmetic mean and 107x in median. Our results highlight the importance of probabilistic upper bounds in \ToolName{}. 

\subsubsection{Minimum Uniconstraint Count}

In Section \ref{section:methodologycounting}, we placed an upper bound on the count of inputs that reach the target along an execution path for a given input region using the minimum uniconstraint count. In this experiment, we compare our technique which uses information from a single uniconstraint count with an alternate one that incorporates information from all uniconstraint counts by multiplying them. 
\begin{table}[t!]
     \caption{\label{tab:aba_minbranch}\small\textbf{ Mean time to trigger the bug across various techniques to approximate path counts over 10 trials.}}
         \vspace{-0.2cm}
    \centering
    \footnotesize
\begin{tabular}{lcc}
\toprule
\multicolumn{1}{c}{\textbf{Bug ID}}  & \multicolumn{1}{c}{\textbf{\ToolName{}}}   & \multicolumn{1}{c}{\textbf{Multiply Uniconstraint Counts}}  \\
\midrule
\texttt{XML009}  & 13s    & 5m59s      \\
\texttt{PNG003}   & 15s    & 6m15s     \\
\texttt{XML017}   & 16s   & 1m31s      \\
\texttt{PHP004}    & 1m04s  & 13m01s  \\
\texttt{PDF011}   & 1m41s & 34m31s    \\
\texttt{PHP009}    & 1m07s  & 7m16s   \\
\texttt{SSL020}  & 9m16s   & 1h10m26s \\
\texttt{TIF009}    & 9m49s  & 19m09s  \\
\texttt{PDF008}   & 3m21s  & 56m57s   \\
\midrule
\multicolumn{2}{c|}{Arithmetic mean speedup } & 13x \\
\multicolumn{2}{c|}{Median speedup} & 12x \\
\bottomrule
\end{tabular}
\vspace{-0.3cm}
\end{table}

Table \ref{tab:aba_minbranch} summarizes the results. \ToolName{} improves upon the multiply uniconstraint counts technique by 13x in arithmetic mean and 12x in median, showing the utility of approximating the count along an execution path using the minimum uniconstraint count. We hypothesize this improvement occurs because multiplying uniconstraint counts to approximate the count along a path corresponds to an independence assumption between individual constraints (i.e., the branch constraints share no variables and hence the counts are independent), which is generally not true for most real-world programs, as shown in the symbolic execution literature~\cite{dart, klee, Godefroid2008AutomatedWF, qsym, driller, savior}.

\subsubsection{Path Selection}
We discuss in Section \ref{section:methodologycounting} the importance of selecting alternate paths with large counts due to approximation error, leading us to use the uncertainty term from the multi-armed bandit literature~\cite{multiarmed} in Algorithm \ref{alg:noisycumulativecountoracle}. In this experiment, we compare against alternate strategies based on the multi-armed bandit literature. We compare against a strategy that sets the uncertainty term to zero and greedily picks the path with the largest count (Greedy). We also compare against a variant called Epsilon-greedy that also sets the uncertainty term to zero but instead of following Greedy all the time, it randomly selects another path based on a coin flip with bias $\epsilon$, set to $\epsilon=0.5$ to equally balance the trade-off.

\begin{table}[t]
     \caption{\label{tab:aba_path}\small\textbf{ Mean time to trigger the bug across various techniques for path selection over 10 trials.}}
         \vspace{-0.2cm}
    \centering
    \footnotesize
\begin{tabular}{lccc}
\toprule
\multicolumn{1}{c}{\textbf{Bug ID}}     & \multicolumn{1}{c}{\textbf{\ToolName{}}}& \multicolumn{1}{c}{\textbf{Epsilon-greedy}} & \multicolumn{1}{c}{\textbf{Greedy}} \\
\midrule
\texttt{XML009}    & 13s   & 9s     & 11s      \\
\texttt{PNG003}    & 15s   & 38s    & 23s      \\
\texttt{XML017}    & 16s   & 11s    & 13s      \\
\texttt{PHP004}    & 1m04s & 2m08s  & 3m12s    \\
\texttt{PDF011}    & 1m41s & 2m06s  & 1m03s    \\
\texttt{PHP009}    & 1m07s & 2m14s  & 3m21s    \\
\texttt{SSL020}    & 9m16s & 14m50s & 4h38m    \\
\texttt{TIF009}    & 9m49s & 15m42s & 4h54m30s \\
\texttt{PDF008}    & 3m21s & 5m22s  & 1h40m30s \\
\midrule
\multicolumn{2}{c|}{Arithmetic mean speedup } & 1.5x & 11x \\
\multicolumn{2}{c|}{Median speedup} & 1.6x & 3x \\
\bottomrule
\end{tabular}
    \vspace{-3mm}
\end{table}
Table \ref{tab:aba_path} summarizes the results. While \ToolName{} improves upon Greedy by 11x on average and 3x in median, it only improves upon Epsilon-greedy by 1.5x on average and 1.6x in median. Our results show the utility of selecting alternate paths to reflect our uncertainty, but also indicate that simple strategies such as Epsilon-greedy can work as well as more advanced ones that incorporate an uncertainty correction factor. 
\vspace{0.3cm}

\begin{longfbox}
\textbf{Result 4:} 
Our experimental results justify \ToolName{}'s design choices with speedups  $\geq1.5x$ in arithmetic mean and $\geq1.6x$ in median. 
\end{longfbox}
\vspace{-0.5cm}

\section{Related Work}

\noindent\textbf{Approximate Counting.}
Approximate counting has been used in many different contexts including counting the number of solutions to SAT formulas~\cite{karpdnf, kuldeep}, flash memory~\cite{approximatecountingflashmemory}, and database systems~\cite{approximatecountingdatabase}. Techniques for approximate counting build upon Monte Carlo counting as well as universal hash functions \cite{kuldeep}, which provide the property of uniformly partitioning each object to be counted into roughly equally-sized groups. We plan to investigate incorporating such techniques in the future. 

Recently, approximate counting was also used in seed scheduling for coverage-guided fuzzing. She et al. approximate the count of reachable and feasible edges using graph centrality~\cite{kscheduler}. In contrast, we approximate the count of inputs that reach the target using Monte Carlo counting for directed greybox fuzzing. Generalizing \ToolName{} from directed greybox fuzzing to the coverage-guided fuzzer setting remains an open question for future work and potentially may involve information entropy from B{\"{o}}hme et al.~\cite{entropic} or abstraction functions from Salls et al.~\cite{abstractionfunctions}.

\noindent\textbf{Directed Greybox Fuzzing. }
Starting with the promising results of \texttt{AFLGo}: finding the HeartBleed vulnerability orders-of-magnitude faster than a directed whitebox fuzzer~\cite{aflgo}, directed greybox fuzzing has seen multiple research directions. One line of work incorporates additional information into the distance computations such as branch distance~\cite{cafl} or function similarity~\cite{hawkeye}. In contrast, \ToolName{} uses noisy binary search and approximate counts, not distance, to guide the fuzzer. 

Based on the observation that directed greybox fuzzers consume a lot of time on executions that fail to reach the target, another promising line of work seeks to increase the fuzzer's efficiency by not executing on inputs that are either unlikely to reach the target~\cite{fuzzguard} or provably cannot~\cite{beacon}. Our approach is complementary to such techniques as we can potentially use them to bias our random input selection process to avoid such inputs. Recent work has also directed a fuzzer with application-specific techniques~\cite{parmesan, regressiongreybox, androiddirected, lineartemporallogic} and incorporating such application-specific techniques is an interesting question for future work. 
\vspace{-0.2cm}

\section{Conclusion}
In this paper, we build an asymptotically optimal directed greybox fuzzer using noisy binary search and a noisy counting oracle.  We also empirically show the promise of our fuzzer as it outperforms existing directed greybox fuzzers by up to two orders of magnitude, on average, over Magma and Fuzzer Test Suite.

\section*{Acknowledgements} 

We thank Clayton Sanford, Samuel Deng, Andreas Kellas, Amol Pasarkar, Dennis Roelke, Gabriel Ryan, Zhongtian Chen, Yuhao Li, Ming Yuan, Christian Kroer, and Junfeng Yang for their helpful comments, and the reviewers for their valuable feedback. 
Peter Coffman helped create tables, improve code quality, and optimize the implementation. 
Abhishek Shah is supported by an NSF Graduate Fellowship. 
This work is supported partially by NSF grants CNS-18-42456, CNS-18-01426; a NSF CAREER award; a Google Faculty Fellowship; a JP Morgan Faculty Fellowship; a Capital One Research Grant; and an Institute of Information \& Communications Technology Planning \& Evaluation (IITP) grant funded by the Korea Government (MSIT)
(No.2020-0-00153).
Any opinions, findings, conclusions, or recommendations expressed herein are those of the authors, and do not necessarily reflect those of the US Government, NSF, Google, Capital One, J.P. Morgan, or the Korean Government.

\bibliographystyle{ACM-Reference-Format}
\balance
\bibliography{ref}
\appendix

\section{Proofs}
\label{appendix:proofs}
\begin{proof}[Proof of Theorem~\ref{theorem:lowerbound}]
Information theory~\cite{infotheorybits} states that to identify (i.e., encode) a unique element in a set containing $N$ elements, we require at least $\log(N)$ bits. Similarly, to identify a target-reaching input in an input space of size $N$, any directed fuzzing algorithm requires at least $\log(N)$ oracle queries, up to constant factors, since each oracle query provides a constant $c$ bits of information.
\end{proof}

\begin{proof}[Proof of Theorem~\ref{theorem:complexity}]

Ben-Or et al. ~\cite{noisybinarysearchorr} in Theorem 2.1 prove that their noisy binary search algorithm requires $(1-\delta)*\frac{\log(N)}{(\frac{1}{2} - p)^2}$ comparisons in expectation to find the target with success probability at least $1-\delta$. If we map our noisy oracle queries to their noisy comparisons and our input space with a lexicographic total order to their array, our fuzzing algorithm in Algorithm~\ref{alg:noisyfuzzer} directly translates to their noisy binary search algorithm and therefore inherits the same analysis. The algorithm analysis uses information entropy arguments to show that the expected information gain increases at each query, followed by concentration bounds to show that when the algorithm terminates, the algorithm can identify the region containing the target with high probability. For more details, see Section 2.3 in \cite{noisybinarysearchorr}.

\end{proof}

\begin{proof}[Proof of Theorem~\ref{theorem:optimality}]
Theorem 2.8 by Ben-Or et al. ~\cite{noisybinarysearchorr} shows that  $\Omega((1-\delta)*\frac{\log(N)}{(\frac{1}{2} - p)^2})$ is a lower bound for any noisy binary search algorithm (i.e., cannot be improved upon) and therefore implies that Algorithm~\ref{alg:noisyfuzzer} is optimal. This lower bound results from a reduction to a well-studied problem in information theory where two parties wish to communicate over a noisy channel (i.e., noisy channel coding problem). For exact details, we refer the reader to Section 2.4 in ~\cite{noisybinarysearchorr}.
We note in the special case of a noisy counting oracle that returns $c=1$ bit of information without noise ($p=0$), Algorithm~\ref{alg:noisyfuzzer} also meets the lower bound, up to constant factors, from Theorem~\ref{theorem:lowerbound}, so it is optimal in both noisy and noiseless settings. 
\end{proof}

\section{Total Order Assignment}
\label{appendix:totalorder}
In Section ~\ref{section:implementation}, we discussed that although our binary search algorithm is agnostic to the underlying total order, which is necessary to ensure splitting an input region is unambiguous, lexicographic order is a poor choice because it assumes that all bytes equally contribute to the count of inputs reaching the target. Empirical evidence from prior work~\cite{neutaint, neuzz, redqueen, greyone, hotbytes} has shown that not all input bytes equally contribute to program behaviors and therefore such an assumption does not hold for many real-world programs. In this experiment, we show that lexicographic order is poor choice by comparing it with our technique to assign a total order.

\begin{table}[!b]
     \caption{\label{tab:aba_ordering}\small\textbf{ Mean time to trigger the bug with and without lexicographic order across 10 trials.}}
    \centering
    \footnotesize
\begin{tabular}{lcc}
\toprule
\multicolumn{1}{c}{\textbf{Bug ID}} & \multicolumn{1}{c}{\textbf{\ToolName{}}}     & \multicolumn{1}{c}{\textbf{Lexicographic Order}}  \\
\midrule
\texttt{XML009} &   13s  &         8m14s \\
\texttt{PNG003} &   15s  &        13m27s  \\
\texttt{XML017} &   16s  &         4m12s  \\
\texttt{PHP004}  & 1m04s &      4h12m21s  \\
\texttt{PDF011}  & 1m41s  &   $T.O^\star$ \\
\texttt{PHP009} & 1m07s   &      3h30m17s \\
\texttt{SSL020}  & 9m16s  &        20m18s \\
\texttt{TIF009} & 9m49s  &        16m23s  \\
\texttt{PDF008}  & 3m21s &   $T.O^\star$  \\
\midrule
\multicolumn{2}{c|}{Arithmetic mean speedup } & 210x \\
\multicolumn{2}{c|}{Median speedup} & 54x \\
\bottomrule
\end{tabular}

\end{table}
Table \ref{tab:aba_ordering} summarizes the results. \ToolName{} outperforms the lexicographic ordering by 210x in arithmetic mean and 54x in median, showing that lexicographic ordering is a poor choice. In the future, we plan to investigate what properties constitute an optimal total order assignment.

\section{Preprocessing Times}
\label{appendix:preprocessing}

\begin{table}[h]
     \caption{\label{tab:preprocessingmagma}\small\textbf{Mean preprocessing times over 10 trials for the Magma and Fuzzer Test Suite benchmarks. }}
   \footnotesize
    \centering
\begin{tabular}{lcccr}
\toprule
\multicolumn{1}{c}{\textbf{Library}}    & \multicolumn{1}{c}{\textbf{\ToolName{}}} & \multicolumn{1}{c}{\textbf{AFLGo}} & \multicolumn{1}{c}{\textbf{ParmeSan}} &
\multicolumn{1}{c}{\textbf{CFG Nodes}} \\
\midrule
 \texttt{libpng (Magma)}     & 54s   & 1m52s     & 32s     & 6940   \\
 \texttt{libtiff (Magma)}    & 55s   & 10m39s    & 33s     & 15485  \\
 \texttt{libxml2 (Magma)}    & 2m41s & 24m08s    & 8m18s   & 65735  \\
 \texttt{openssl (Magma)}    & 5m11s & 1h31m11s  & 58m15s  & 95949\\
\texttt{php (Magma)}         & 3m21s & 14h20m09s & N/A     & 371648  \\
\texttt{poppler (Magma)}     & 3m27s & 2h28m09s  & 2m26s   & 71591  \\
\texttt{libjpeg (FTS)}       & 7s    & 1m45s     & 32s     & 11173  \\
\texttt{libpng (FTS)}        & 38s   & 54s       & 31s     & 5257   \\
 \texttt{freetype2 (FTS)}    & 1m15s & 12m07s    & 38s     & 28662  \\
\midrule
        \multicolumn{1}{c|}{Arithmetic mean  } & 2m04s & 2h07m53s & 8m58s & 74716 \\
        \multicolumn{1}{c|}{Median } & 1m15s & 12m07s & 35s & 28662 \\
\bottomrule
\end{tabular}

\end{table}

Existing directed greybox fuzzers use preprocessing to better identify which inputs are more likely to reach the target as described in Section \ref{section:implementation}. We measure the preprocessing times of the tested fuzzers to see how they compare. 

For \texttt{AFLGo}, we measure the time it takes to compute distance over the control-flow graph, which consists of visiting every function, computing intra-function distances, and using callgraphs to compute distances between functions. For \texttt{ParmeSan}, it uses a dynamic CFG, so it is difficult to accurately measure this time since preprocessing is conflated with runtime. We instead approximate this time by measuring the time it takes to run over only the initial seed corpus, in which the dynamic CFG is constructed and distances are computed. We emailed the authors to ensure our setup was reasonable and they confirmed that our experimental setup is reasonable given the dynamic CFG component. For \ToolName{}, we measure the time it takes to perform preprocessing as described in Section \ref{section:implementation}. Table~\ref{tab:preprocessingmagma} summarizes the results for both Magma and Fuzzer Test Suite for the bugs targets found in Section~\ref{sec:evaluationrq1}.  
\section{Monte Carlo Execution Exceptions}
\label{appendix:crashrate}

\begin{table}[t!]
     \caption{\label{tab:crashrate}\small\textbf{ Average proportion of Monte Carlo Executions with Exceptions on the Magma benchmark over 10 trials.}}
    \centering
    \footnotesize
\begin{tabular}{cc}
\toprule
\multicolumn{1}{c}{\textbf{Library}}    & \multicolumn{1}{c}{\textbf{Proportion of Executions with Exceptions  (\%)}}  \\
\midrule
\texttt{libpng}     & 0.26\%    \\
\texttt{libtiff}    & 2.31\%    \\
\texttt{libxml2}    & 0.84\%     \\
\texttt{openssl}    & 1.22\%    \\
\texttt{php}        & 4.04\%     \\
\texttt{poppler}    & 2.24\%    \\
\texttt{sqlite3}    & 10.9\%    \\
\midrule
\multicolumn{1}{c|}{Arithmetic mean  } & 3.11\% \\
\multicolumn{1}{c|}{Median } & 2.14\% \\
\bottomrule
\end{tabular}
\end{table}

Since handling a large number of program exceptions can potentially incur high overheads (i.e., context switches from signal handling), in this experiment, we investigate how many times \Execution \ handles program exceptions. Specifically, we measure the ratio between the number of executions which require \Execution \ to handle program exceptions to the total number of executions in our Magma evaluation, repeated 10 times to reduce variability.

Table ~\ref{tab:crashrate} summarizes the results, with $3.11\%$ in arithmetic mean and $2.14\%$ in median for the proportion. This experiment shows that many Monte Carlo Executions do not involve program exceptions and therefore incur low overhead, a finding that better helps explain our speedups. 
\begin{algorithm}[!]
\footnotesize
\caption{\small Optimal Deterministic Fuzzer.} 
\label{alg:fuzzer} 
\lstset{basicstyle=\ttfamily\footnotesize, breaklines=true}
\begin{tabular}{|lp{2.6in}|}\hline
\textbf{Input}:
    & \textit{$\mathbb I$} $\leftarrow$ Input Space as Array\\
    & \textit{$O$} $\leftarrow$ Noiseless ($p=0$) Counting Oracle from Equation~\ref{equation:oracle} \\
\hline
\end{tabular}
\begin{algorithmic}[1]
\State $l = 0; r = |\mathbb I| - 1$ \Comment{\textcolor{purple}{initialize left and right bounds of $\mathbb I$}}

\While{l < r}
    \State $m =\lfloor (l+r)/2\rfloor$ \Comment{\textcolor{purple}{select midpoint input}}
        
    \State $I_L, I_R = \textsf{left and right input regions of index m}$
    \If{ $O(I_L, I_R) = 1 $} 
        \State  $r = m$ \Comment{\textcolor{purple}{select left subregion}}
    \Else
        \State $l = m + 1$ \Comment{\textcolor{purple}{select right subregion}}
    \EndIf
\EndWhile
\end{algorithmic}
\end{algorithm}

\begin{table}[!]
     \caption{\label{tab:magma4}\small\textbf{ 
      Mean time to trigger Magma bugs for each tested fuzzer's undirected counterpart over 20 trials. Since \texttt{Angora} crashed on \texttt{php}, we write N/A for it. See Table~\ref{tab:magma} for the full caption. }}
   \footnotesize
    \centering
    \setlength{\tabcolsep}{2.5pt}
\begin{tabular}{|l|c|ccr|ccr|}
\hline
\rowcolor[HTML]{EFEFEF}
 & \multicolumn{1}{c|}{\textbf{\ToolName{}}} & \multicolumn{3}{c|}{\textbf{\texttt{AFL}}} & \multicolumn{3}{c|}{\textbf{\texttt{Angora}}}  \\
\rowcolor[HTML]{EFEFEF} \multirow{-2}{*}{\textbf{Bug ID}}
 & \multicolumn{1}{c|}{\textbf{Time}} & \multicolumn{1}{c}{\textbf{Time}}& \multicolumn{1}{c}{\textbf{$(x)$}} &\multicolumn{1}{c|}{\textbf{$(p)$}} & \multicolumn{1}{c}{\textbf{Time}} &  \multicolumn{1}{c}{\textbf{$(x)$}} & \multicolumn{1}{c|}{\textbf{$(p)$}}  \\
\hline
\texttt{PDF010} &    3m15s &    4m25s &          1x &     0.09  &     $T.O^\star$ &        $>$111x &       <0.01   \\
\texttt{PDF016} &    3m23s &    4m02s &          1x &     0.11  &      12m    &             4x &       <0.01   \\
\texttt{PHP004} &    1m04s &    2m39s &          2x &     <0.01  &      N/A    &         N/A    &          N/A  \\
\texttt{PHP009} &    1m07s &    3m38s &          3x &     <0.01  &      N/A    &         N/A    &          N/A  \\
\texttt{PHP011} &    1m01s &    2m29s &          2x &     <0.01  &      N/A    &         N/A    &          N/A  \\
\texttt{PNG003} &      15s &      15s &          1x &     0.25  &         59s &             4x &       <0.01   \\
\texttt{PNG006} &    1m36s &  $T.O^\star$ &     $>$225x &     <0.01  &       2m42s &             2x &       0.03   \\
\texttt{SSL002} &    1m44s &    4m06s &          2x &     <0.01  &      55m53s &            32x &       <0.01   \\
\texttt{SSL003} &    1m39s &    2m50s &          2x &     <0.01  &      20m40s &            13x &       <0.01   \\
\texttt{SSL009} &    4m59s &  $T.O^\star$ &      $>$72x &     <0.01  &    4h56m06s &            59x &       <0.01   \\
\texttt{TIF005} &    9m33s &  $T.O^\star$ &      $>$38x &     <0.01  &    3h57m57s &            25x &       <0.01   \\
\texttt{TIF006} &    9m36s & 5h21m19s &         33x &     <0.01  &    4h40m09s &            29x &       <0.01   \\
\texttt{TIF007} &    8m18s & 1h13m28s &          9x &     0.04  &    1h31m48s &            11x &       <0.01   \\
\texttt{TIF012} &    9m59s & 1h54m37s &         11x &     <0.01  &    4h51m22s &            29x &       <0.01   \\
\texttt{TIF014} &    1m36s & 4h55m49s &        185x &     <0.01  &    5h38m17s &           211x &       <0.01   \\
\texttt{XML017} &      16s &    1m23s &          5x &     <0.01  &       1m57s &             7x &       <0.01   \\
\rowcolor[HTML]{ECF4FF}
\texttt{PDF003} &    1m39s &  $T.O^\star$ &     $>$218x &     <0.01  &     $T.O^\star$ &        $>$218x &       <0.01   \\
\rowcolor[HTML]{ECF4FF}
\texttt{PDF008} &    3m21s &  $T.O^\star$ &     $>$107x &     <0.01  &     $T.O^\star$ &        $>$107x &       <0.01   \\
\rowcolor[HTML]{ECF4FF}
\texttt{PDF011} &    1m41s &  $T.O^\star$ &     $>$214x &     <0.01  &     $T.O^\star$ &        $>$214x &       <0.01   \\
\rowcolor[HTML]{ECF4FF}
\texttt{PDF018} &    1m43s &  $T.O^\star$ &     $>$210x &     <0.01  &     $T.O^\star$ &        $>$210x &       <0.01   \\
\rowcolor[HTML]{ECF4FF}
\texttt{PDF019} &    1m37s &  $T.O^\star$ &     $>$223x &     <0.01  &     $T.O^\star$ &        $>$223x &       <0.01   \\
\rowcolor[HTML]{ECF4FF}
\texttt{PNG001} &    3m17s &  $T.O^\star$ &     $>$110x &     <0.01  &     $T.O^\star$ &        $>$110x &       <0.01   \\
\rowcolor[HTML]{ECF4FF}
\texttt{PNG007} &    3m21s &  $T.O^\star$ &     $>$107x &     <0.01  &     $T.O^\star$ &        $>$107x &       <0.01   \\
\rowcolor[HTML]{ECF4FF}
\texttt{SSL020} &    9m16s &  $T.O^\star$ &      $>$39x &     <0.01  &     $T.O^\star$ &         $>$39x &       <0.01   \\
\rowcolor[HTML]{ECF4FF}
\texttt{TIF001} &    9m43s &  $T.O^\star$ &      $>$37x &     <0.01  &     $T.O^\star$ &         $>$37x &       <0.01   \\
\rowcolor[HTML]{ECF4FF}
\texttt{TIF002} &    9m58s &  $T.O^\star$ &      $>$36x &     <0.01  &     $T.O^\star$ &         $>$36x &       <0.01   \\
\rowcolor[HTML]{ECF4FF}
\texttt{TIF009} &    9m49s &  $T.O^\star$ &      $>$37x &     <0.01  &     $T.O^\star$ &         $>$37x &       <0.01   \\
\rowcolor[HTML]{ECF4FF}
\texttt{XML009} &      13s &  $T.O^\star$ &    $>$1662x &     <0.01  &     $T.O^\star$ &       $>$1662x &       <0.01   \\
\hline
\multicolumn{2}{|c|}{\textbf{Mean speedup}}   & & 128x & & & 142x & \\
\multicolumn{2}{|c|}{\textbf{Median speedup}} & & 37x  & & & 37x  & \\

\hline
\end{tabular}

\end{table}

\begin{table}[!]
     \caption{ \label{tab:fts4} \small\textbf{ 
      Mean time to reach Fuzzer Test Suite targets for each tested fuzzer's undirected counterpart over 20 trials. See Table~\ref{tab:fts} for the full caption.}} 
   \footnotesize
    \centering
    \setlength{\tabcolsep}{1.5pt}
\begin{tabular}{|l|r|ccr|ccr|}
\hline
\rowcolor[HTML]{EFEFEF}
 & \multicolumn{1}{c|}{\textbf{\ToolName{}}} & \multicolumn{3}{c|}{\textbf{\texttt{AFL}}} & \multicolumn{3}{c|}{\textbf{\texttt{Angora}}}  \\
\rowcolor[HTML]{EFEFEF} \multirow{-2}{*}{\textbf{Bug ID}}
 & \multicolumn{1}{c|}{\textbf{Time}} & \multicolumn{1}{c}{\textbf{Time}}& \multicolumn{1}{c}{\textbf{$(x)$}} &\multicolumn{1}{c|}{\textbf{$(p)$}} & \multicolumn{1}{c}{\textbf{Time}} &  \multicolumn{1}{c}{\textbf{$(x)$}} & \multicolumn{1}{c|}{\textbf{$(p)$}}  \\
\hline
\texttt{ttgload.c:1710}  &       1s &       1s &          1x &     0.07   &          1s &             1x &      0.07     \\
\texttt{ttinterp.c:2186} &    9m57s &  $T.O^\star$ &      $>$36x &     <0.01  &         24m &             2x &      <0.01    \\
\texttt{cf2intrp.c:361}  &      58s &      40m &         41x &     <0.01  &     $T.O^\star$ &        $>$372x &      <0.01    \\
\texttt{jdmarker.c:659}  &      32s &    1h10m &        131x &     <0.01  &       1h15m &           141x &      <0.01    \\
\texttt{pngrutil.c:139}  &       1s &       1s &          1x &     0.07   &          1s &             1x &      0.07     \\
\texttt{pngrutil.c:3182} &      28s &    3m20s &          7x &     <0.01  &       1m22s &             3x &      <0.01    \\
\texttt{pngread.c:738}   &       1s &       1s &          1x &     0.07   &          1s &             1x &      0.07     \\
\rowcolor[HTML]{ECF4FF}
\texttt{pngrutil.c:1393} &      51s &  $T.O^\star$ &     $>$424x &     <0.01  &     $T.O^\star$ &        $>$424x &      <0.01    \\
\hline

\multicolumn{2}{|c|}{\textbf{Mean speedup}}   & & 80x & & & 118x & \\
\multicolumn{2}{|c|}{\textbf{Median speedup}} & & 22x  & & & 3x  & \\

\hline
\end{tabular}

\end{table}

\end{document}